\documentclass{article}
\usepackage[utf8]{inputenc}
\usepackage{url} 
\usepackage[authoryear]{natbib}
\usepackage[colorlinks=true,linkcolor=blue,citecolor=blue,urlcolor=blue,breaklinks]{hyperref}
\usepackage{cleveref}
\usepackage{setspace}
\usepackage{amsmath,amssymb,amsthm,textcomp}
\usepackage{thmtools, thm-restate}
\usepackage{fullpage}
\usepackage{enumerate}
\usepackage[percent]{overpic}
\usepackage[linesnumbered,ruled,vlined ]{algorithm2e}
\usepackage{custom_tex}
\usepackage{caption}
\usepackage{subcaption}
\usepackage{graphicx} 
\usepackage{grffile}
\usepackage{xfrac}
\usepackage[export]{adjustbox}
\numberwithin{figure}{section}
\theoremstyle{definition}
\newtheorem{theorem}{Theorem}[section]
\newtheorem{prop}[theorem]{Proposition}
\newtheorem{lemma}[theorem]{Lemma}

\newcommand{\sig}{\Sigma}

\newcommand{\R}{\mathbb{R}}
\newcommand{\ep}{\varepsilon}
\newcommand{\E}{\mathbb{E}}
\newcommand{\C}{\textnormal{Cov}}
\newcommand{\Va}{\textnormal{Var}}
\newcommand{\Po}{\textnormal{Poisson}}
\newcommand{\V}{\|}
\newcommand{\Fr}{\textnormal{Fr}}
\newcommand{\I}{\mathcal{I}}
\newcommand{\bitem}{\begin{itemize}}
\newcommand{\eitem}{\end{itemize}}
\newcommand{\benum}{\begin{enumerate}}
\newcommand{\eenum}{\end{enumerate}}
\newcommand{\beq}{\begin{equation}}
\newcommand{\eeq}{\end{equation}}
\newcommand{\beqs}{\begin{equation*}}
\newcommand{\eeqs}{\end{equation*}}

\allowdisplaybreaks
\newcommand*\samethanks[1][\value{footnote}]{\footnotemark[#1]}

\graphicspath{
 {../experiments/xfel/images/}
 {../public_software/experiments/standardize/images/}
 {../public_software/experiments/precise_PT/images/} 
 {../public_software/experiments/random_subspace/images/} 
 {../public_software/experiments/snp_data/images/} 
}

\title{$e$PCA: High Dimensional Exponential Family PCA
}

\author{
	Lydia T. Liu\thanks{The first two authors contributed equally to this work.} \thanks{e-mail: 	ltliu@princeton.edu. 1980 Frist Center, Princeton University, Princeton, NJ, 08544}
	  \and
  Edgar Dobriban\samethanks \thanks{e-mail: dobriban@stanford.edu. Department of Statistics, Stanford University, Stanford, CA, 94305}
	\and
	Amit Singer  \thanks{e-mail: amits@math.princeton.edu. Department of Mathematics, and Program in Applied and Computational Mathematics, Princeton University, Princeton, NJ, 08544}
 }

\date{\today}

\begin{document}
\maketitle

\begin{abstract}
Many applications, such as photon-limited imaging and genomics, involve large datasets with noisy entries from exponential family distributions. It is of interest to estimate the covariance structure and principal components of the noiseless distribution. 
%In photon-limited imaging (e.g. XFEL) we want to estimate the covariance of the pixel intensities of 2-D images, where the pixels are low-intensity Poisson variables. In genomics we want to estimate population structure from biallelic---Binomial(2)---genetic markers such as Single Nucleotide Polymorphisms (SNPs). 
Principal Component Analysis (PCA), the standard method for this setting, can be inefficient when the noise is non-Gaussian.

We develop $e$PCA (exponential family PCA), a new methodology for PCA on exponential family distributions. $e$PCA can be used for dimensionality reduction and denoising of large data matrices.  $e$PCA involves the eigendecomposition of a new covariance matrix estimator, constructed in a simple and deterministic way using moment calculations, shrinkage, and random matrix theory. %$e$PCA is as fast as PCA and is suitable for datasets with multiple types of variables. 

We provide several theoretical justifications for our estimator, including the finite-sample convergence rate, and the Marchenko-Pastur law in high dimensions. %A key step of $e$PCA is \emph{homogenization}, a specific variable weighting. For SNPs, this recovers the widely used Hardy-Weinberg equilibrium (HWE) normalization. We show that homogenization improves the signal strength, providing justification for HWE normalization. 
$e$PCA compares favorably to PCA and various PCA alternatives for exponential families, in simulations as well as in XFEL and SNP data analysis.
An open-source implementation is  \href{http://github.com/lydiatliu/epca/}{available}.

\end{abstract}

%\tableofcontents

%\listoffigures

\section{Introduction}

In many applications we have large collections of data vectors with entries sampled from exponential families (such as Poisson or Binomial). 
This setting arises in image processing, computational biology, and natural language processing, among others.  It is often of interest to reduce the dimensionality and understand the structure of the data. 

The standard method for dimension reduction and denoising of large datasets is Principal Component Analysis (PCA) \cite[e.g.,][]{jolliffe2002principal,anderson1958introduction}. However, PCA is most naturally designed for Gaussian data, and there is no commonly agreed upon extension to non-Gaussian settings such as exponential families \cite[see. e.g.,][Sec. 14.4]{jolliffe2002principal}. While there are several proposals for extending PCA to non-Gaussian distributions, each of them has certain limitations, such as computational intractability for large datasets (see Sec. \ref{rel_work} for a detailed discussion).

We propose the new method {\bf $e$PCA} for PCA of data from exponential families. $e$PCA involves the eigendecomposition of a new covariance matrix estimator. Like usual PCA, it can be used for visualization and denoising of large data matrices.  Moreover, $e$PCA has several appealing properties. First, it is a computationally efficient deterministic algorithm that comprises a small number of basic linear algebraic operations, making it as fast as usual PCA and scalable to ``big'' datasets.  This is in contrast to typical likelihood approaches involving iterative methods such as alternating least squares, the EM algorithm etc., without convergence guarantees. Second, it is a flexible method suitable for datasets with multiple types of variables (such as Poisson, Binomial, and Negative Binomial). Third, it has substantial theoretical justification. We provide finite-sample convergence rates, and a precise high-dimensional analysis building on random matrix theory. Fourth, each step of $e$PCA is interpretable, which can be important to practitioners.

We perform extensive simulations with $e$PCA and show that in several metrics it outperforms usual PCA, PCA after standardization, and PCA alternatives for exponential families (see Sec. \ref{sec:denoising}). We apply $e$PCA to simulated X-ray Free Electron Laser (XFEL) data, where it leads to better denoising---visually and in MSE---than PCA. We also apply $e$PCA to a dataset from the Human Genome Diversity Project (HGDP) measuring Single Nucleotide Polymorphisms, where it leads to a clearer structure than PCA.

$e$PCA is publicly available in an open-source Matlab implementation from \url{github.com/lydiatliu/epca/}. That link also has software to reproduce our computational results.

To motivate our method, we now discuss a few potential application areas.  
 
\subsection{Denoising XFEL diffraction patterns}

\begin{figure}[h]
	\centering
	\begin{subfigure}{.25\textwidth}
		\centering
		\includegraphics[scale=0.2,trim = 20 30 20 20, clip]{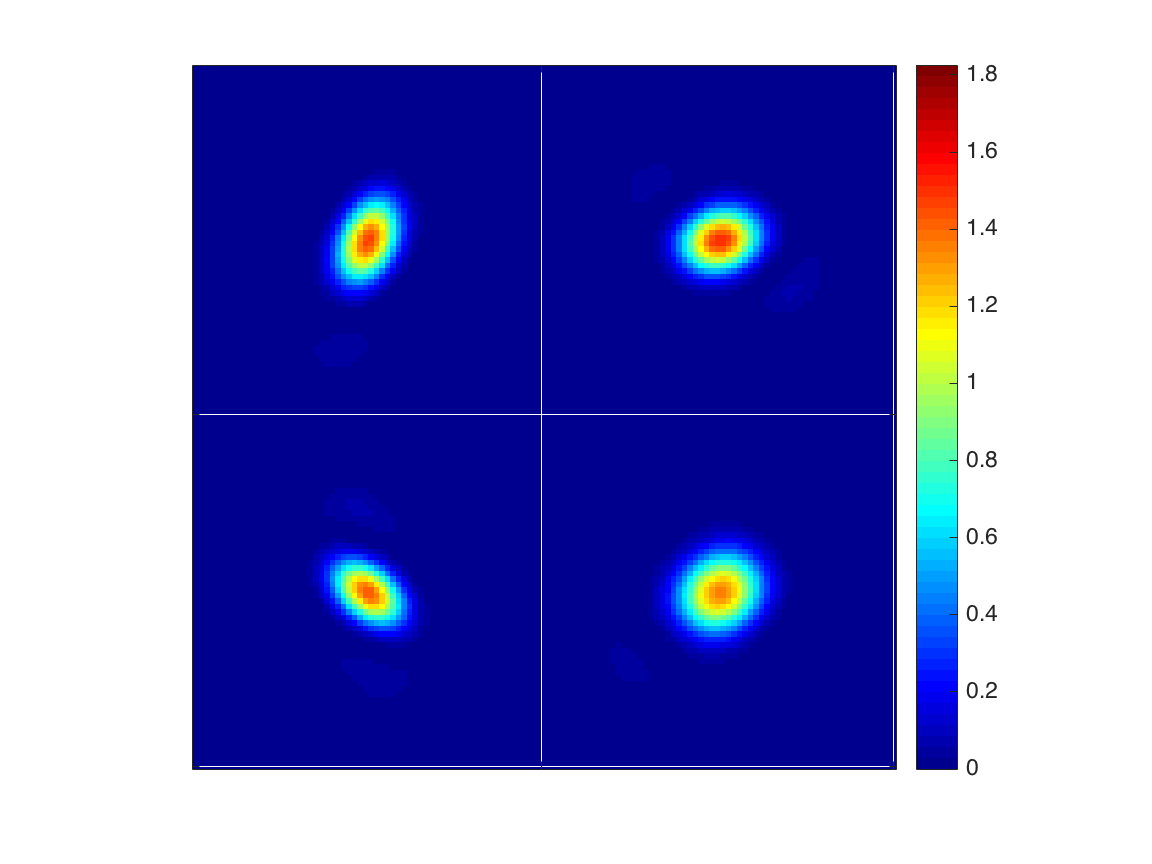}
		\caption{\scriptsize Clean intensity maps}
	\end{subfigure}%
	\begin{subfigure}{.25\textwidth}
		\centering
		\includegraphics[scale=0.2,trim =  20 30 20 20, clip]{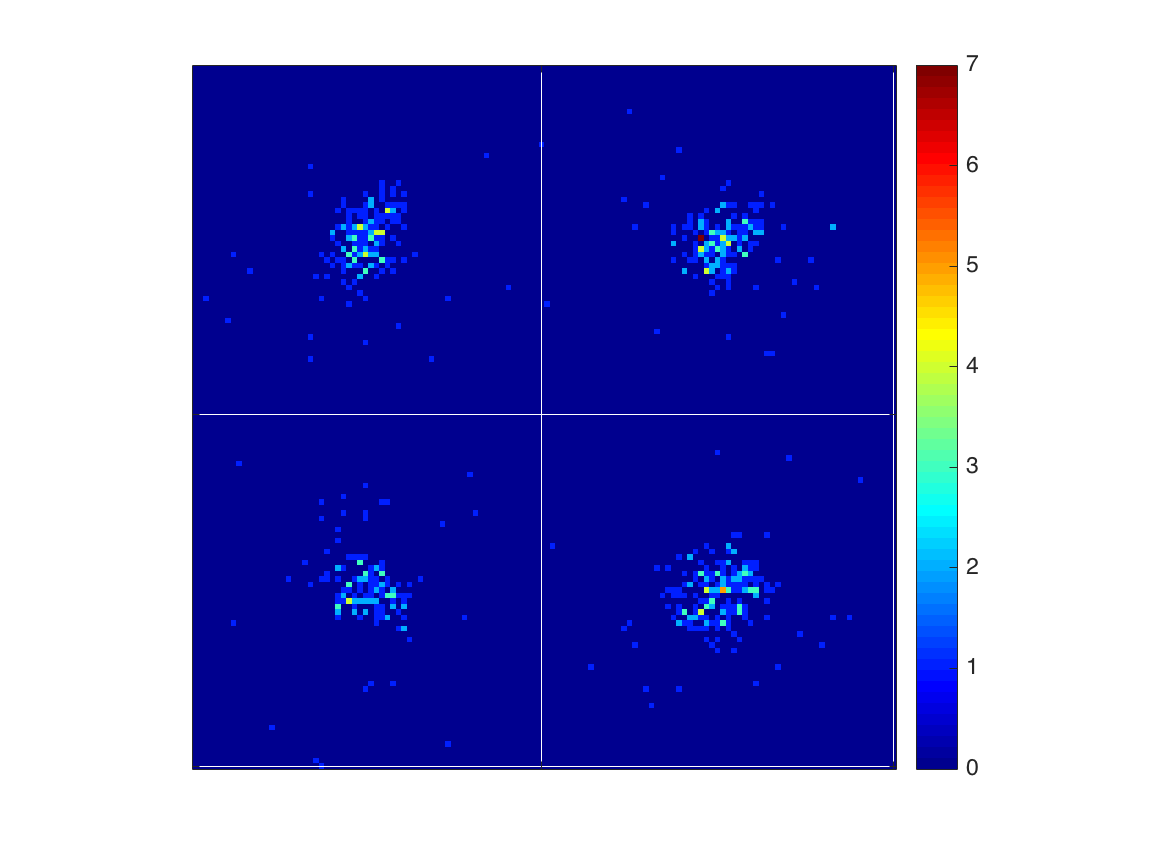}
		\caption{\scriptsize Noisy photon counts}
	\end{subfigure}%
	\begin{subfigure}{.25\textwidth}
		\centering
		\includegraphics[scale=0.2,trim =  20 30 20 20, clip]{{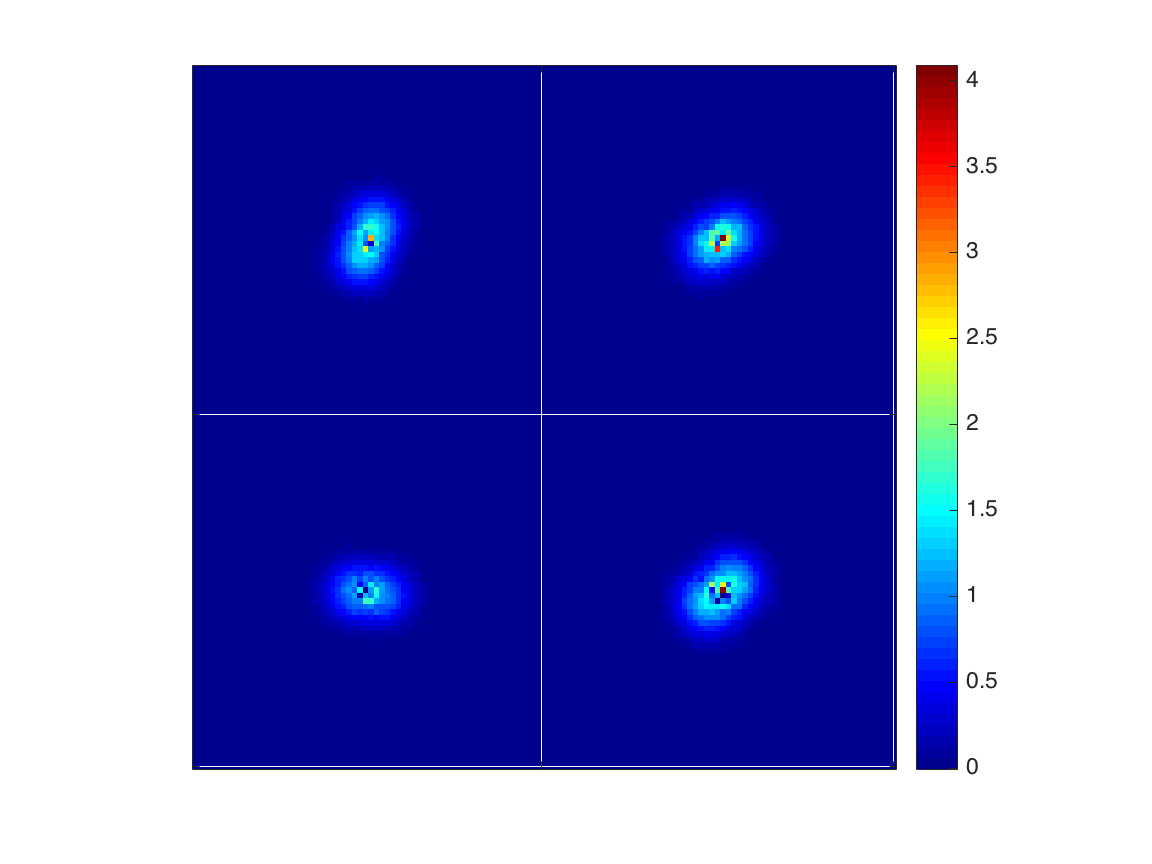}}
		\caption{\scriptsize Denoised (PCA)}
	\end{subfigure}%
	\begin{subfigure}{.25\textwidth}
		\centering
		\includegraphics[scale=0.2,trim =  20 30 20 20, clip]{{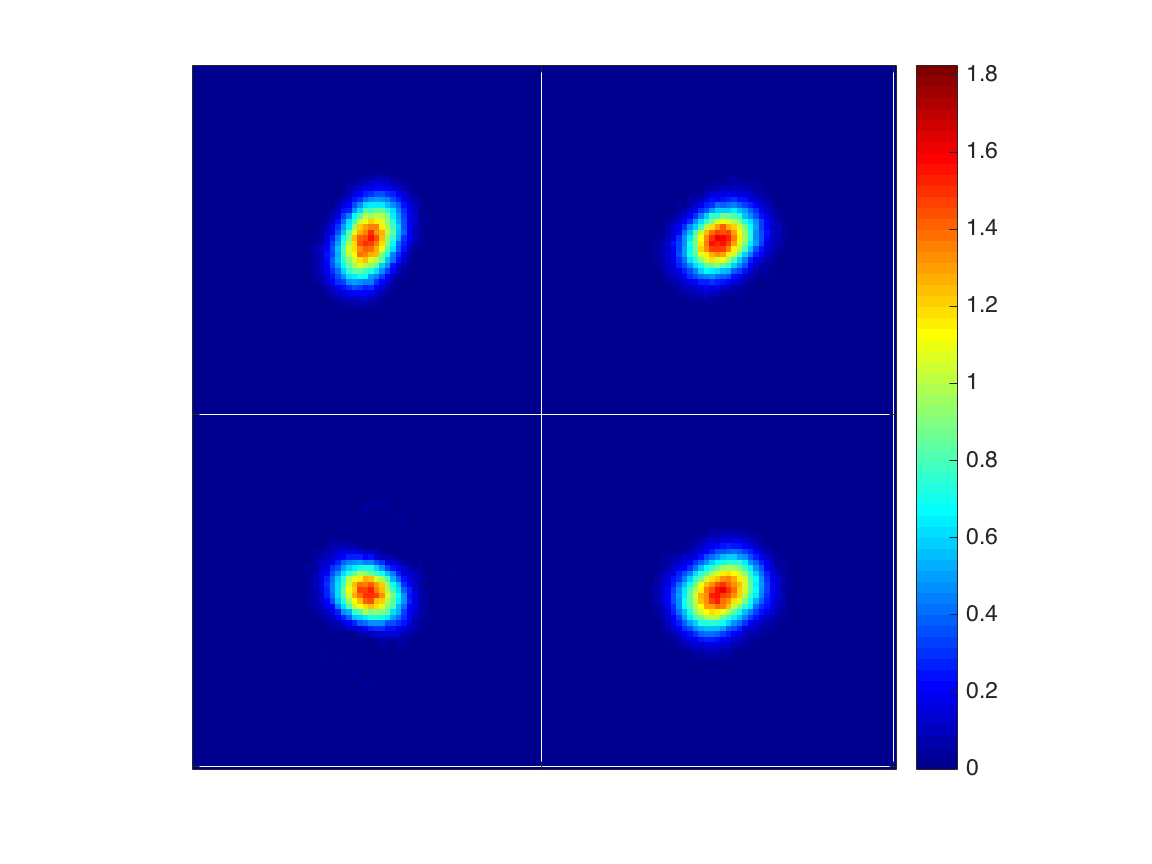}}
		\caption{\scriptsize Denoised ($e$PCA + EBLP)}
	\end{subfigure}
	\caption{XFEL diffraction pattern formation model and denoising. See section \ref{xfel} for details.}
	\label{fig:xfel}
\end{figure}

X-ray free electron lasers (XFEL) are an increasingly popular experimental technique to understand the three-dimensional structure of molecules \citep[e.g.,][]{favre2015xtop, maia2016trickle}.  XFEL imaging leads to two-dimensional diffraction patterns of single particles.  A key advantage of XFEL is that it uses extremely short femtosecond X-ray pulses, during which the molecule does not change its structure. As illustrated in Figure \ref{fig:xfel}, these images are very noisy due to the low number of photons, and the count-noise at each detector follows an approximately Poisson distribution. Further, we only capture one diffraction pattern per particle, and the particle orientations are unknown.

In order to reconstruct the 3-D structure of the particle, one approach is to use expectation-maximization (EM) \citep[e.g.,][]{scheres2007disentangling, duane2009}. Alternatively, assuming that the orientations are uniformly distributed over the special orthogonal group $SO(3)$, Kam's method \citep{kam1977determination, Kam1980} could provide a way to estimate the 3-D structure without optimizing likelihood via EM \citep[see e.g.,][]{Saldin2009}. 
A requirement is to estimate the covariance matrix of the noiseless 2-D images. This motivates us to develop the $e$PCA method for covariance estimation and PCA of Poisson data, and more generally for exponential families. To illustrate the improvement in covariance estimation of $e$PCA over PCA, in Figure \ref{fig:xfel} we show the result of denoising simulated XFEL using different estimated covariance matrices, where EBLP refers to the denoiser we develop in section \ref{denoise} for use in conjunction with $e$PCA.

\subsection{Genetic polymorphism data/SNPs}
In genomics, Single Nucleotide Polymorphism (SNP) data are the basis of thousands of Genome-Wide Association Studies (GWAS), which have recently led to hundreds of novel associations between common traits and genetic variants \citep[e.g.,][]{visscher2012five}.

SNP data can be represented as an $n \times p$ matrix $X$ with $X_{ij}$ equal to the number of minor alleles (0, 1 or 2) of the $j$-th SNP in the genome of the $i$-th individual. The number of individuals $n$ can be more than 10,000, while the number of SNPs can be as large as 2.5 million. Binomial models are natural for such data.  PCA is commonly used to infer population structure from SNP data, with a wide range of applications, including correcting for confounding in GWAS (see e.g., \cite{patterson2006population}). It is thus of interest to understand the proper way to estimate the covariance matrix and PCs.

Among other potential application areas, we point out RNA-sequencing, where negative binomial models are routinely in use \citep{anders2010differential}.

\subsection{Our contributions}
\label{contrib}
We now briefly summarize our contributions:
\benum
\item We propose the new method $e$PCA for PCA of exponential family data. $e$PCA is based on a new covariance estimator that we develop in a sequence of steps (Sec. \ref{covar_norotat} and \ref{cov_est}).  We start with a \emph{diagonal debiasing} of the sample covariance matrix (Sec \ref{diag_deb}), and characterize the finite-sample convergence rate to the population  covariance matrix (Sec. \ref{sec_rate_conv}). 

\item To improve performance for high-dimensional data, we propose a method of \emph{homogenization}, \emph{shrinkage}, and \emph{heterogenization} of the debiased covariance matrix (Sec. \ref{cov_est}). Homogenization is a form of variable weighting, different from the usual method of standardizing the features to have unit variance. We justify it by proving the standard Marchenko-Pastur law \citep{marchenko1967distribution} for the homogenized sample covariance matrix (Sec. \ref{mp_law}), and by showing that homogenization improves the signal strength (Sec. \ref{SNR}). An additional eigenvalue shrinkage step---that we call \emph{scaling}---is needed beyond  the well-understood shrinkage methods for homoskedastic Gaussian distributions \citep[e.g.,][]{Donoho2013}. This leads to our final covariance estimator, and $e$PCA consists of its the eigendecomposition.

We evaluate our covariance estimators in a simulation study, and show that they reduce the MSE for covariance, eigenvalue, and eigenvector estimation (Sec. \ref{bc_ev}). 

\item  For biallelic genetic markers such as Single Nucleotide Polymorphisms (SNPs), homogenization agrees with the widely used normalization assuming Hardy-Weinberg equilibrium (HWE) (Sec. \ref{white_hwe}). This provides perhaps the first theoretical justification for HWE normalization.

\item We apply $e$PCA to develop a new denoising method (Sec. \ref{denoise}), a form of empirical Best Linear Predictor (EBLP) from random effects models \citep[][Sec. 7.4]{searle2009variance}, where we use our covariance estimator to estimate parameters in the BLP denoiser. In areas such as electrical engineering and signal processing, the BLP is known as the ``Wiener filter'' or the ``Linear Minimum Mean Squared Estimator (LMMSE)''  \citep[e.g.,][Ch. 12]{kay1993fundamentals}.

\item We apply $e$PCA denoising to simulated XFEL data where it leads to better denoising than PCA (Sec. \ref{xfel}) 
We also apply $e$PCA to a SNP dataset from the Human Genome Diversity Project (HGDP) \citep{li2008worldwide}, where it leads to a clearer structure in the PC scores than PCA (Sec. \ref{hgdp}).
\eenum
 
\section{Related work}
\label{rel_work}

To give context for our method, we review related work. The reader intersted in the methodology can skip directly to Section \ref{covar_norotat}. 
We refer to \cite{jolliffe2002principal} for a detailed overview of PCA methodology, to \cite{anderson1958introduction} for a more general overview of multivariate statistical analysis including PCA, and to \cite{yao2015large} for discussions of high-dimensional statistics, random matrix theory and PCA. 

\subsection{Standardization and weighting in PCA}

In applying PCA, a key concern is whether or not to standardize the variables \cite[e.g.,][Sec. 2.3]{jolliffe2002principal}. Standardization ensures that results for different sets of random variables are more comparable, and also that PCs are less dominated by individual variables with large variances. Not standardizing makes statistical inference more convenient. In exploratory analyses, however, standardization is usually preferred. In our setting, the homogenization method (Sec. \ref{sec_white}) has several advantages over standardization. 

A more general class of methods is \emph{weighted PCA}, where PCA is applied to rescaled random variables $w_j X(j)$, for some $w_j>0$ \citep[][Sec. 2.3., Sec. 14.2.]{jolliffe2002principal} In general, choosing the weights can be nontrivial. %Some of the suggestions beyond dividing by standard errors include dividing by the range or mean of random variables \citep{gower1966some}. 
Our homogenization step of $e$PCA  (Sec. \ref{sec_white}) is a particular weighting method, justified for data from exponential families. In addition to proposing it, we provide several theoretical justifications: the standard Marchenko-Pastur law, and the improvements in signal to noise ratio (SNR) (see Sec. \ref{sec_white}).

\subsection{PCA in non-Gaussian distributions, GLLVMs}

There have been several approaches suggested for extending PCA to non-Gaussian distributions,   see. e.g., \cite{jolliffe2002principal}, Sec. 14.4. One possibility is to use robust estimates of the covariance matrix \cite[see][Sec. 14.4, for references]{jolliffe2002principal}. %Another approach is to change the quadratic objective function maximized by PCA into one that is suitable for non-Gaussian distributions, such as predictive power \citep[e.g.,][]{qian1994principal}. %Yet another class of methods is Independent Component Analysis (ICA), which looks for derived variables that are not just uncorrelated, but also independent \citep{hyvarinen2004independent}. While ICA works for non-Gaussian distributions, it is not aimed at estimating the low-dimensional subspace of the signals, or estimating the underlying covariance matrix.
Another approach assumes that the natural parameter lies in a low dimensional space  \citep{Collins2001}, and then attempts to maximize the log-likelihood. This leads to a non-convex optimization problem for which an alternating maximization method is proposed, without global convergence guarantees. More recently, \citet{Udell2014, udell2016} described a similar generalization of PCA, while \citet{li2010simple} proposed another likelihood-based method, both without global convergence guarantees. Scalable methods include \cite{josse2016bootstrap}, albeit without precise performance guarantees in high dimensions. 

Within factor analysis, generalized linear latent variable models (GLLVMs) model the relationship of an observed variable from a general distribution with unobserved latent variables \citep{knott1999latent, huber2004estimation}. These flexible likelihood-based methods  enable careful modelling and statistical inference for parameters of interest in low-dimensional settings.  However, estimation and inference are computationally challenging, and published examples have at most 10-20 dimensions \citep{huber2004estimation}. In contrast our algorithm is as fast as PCA and we avoid any optimization problems. In addition, we have some understanding of the performance in high dimensions, by connecting to random matrix theory.

%In contrast, our approach avoids high-dimensional optimization problems and provides precise results about the performance of the method by building on random matrix theory.

%In our settings, it is often more reasonable to model the mean parameter of the exponential family---e.g., the clean images in image analysis---as having a low rank structure. For instance, in an XFEL experiment, the images can be thought of as ``low-complexity'' perturbations of the mean image according to the state of the molecule. For this reason, our methods are not directly comparable in theory or in simulations. 

%\subsection{Factor analysis and Generalized Linear Latent Variable Models}

\subsection{Denoising and covariance estimation by singular value shrinkage}

Recently, results from random matrix theory have been used for studying covariance estimation and PCA for Gaussian and rotationally invariant data \citep[e.g.,][]{shabalin2013reconstruction, Donoho2013,nadakuditi2014optshrink}. While the qualitative insights they identify---e.g., the improvements due to eigenvalue shrinkage---are relevant to our setting, the specific results and methods do not apply directly.

The recent work of \cite{bigot2016generalized}  develops a generalized Stein's Unbiased Risk Estimation (SURE) approach for singular value shrinkage denoising of low-rank matrices in exponential families.  However, their shrinkage formulas become numerically intractable for Frobenius norm beyond Gaussian errors, and they instead introduce a heuristic algorithm. Their work is geared towards higher signal-to-noise ratio settings. % in their simulations. %Moreover in their Poisson-noise simulations, they work in a regime where optimal shrinkage reduces to no shrinkage (e.g., their Fig. 5a), possibly because the signal is very strong. In contrast, we work in a setting of much weaker signal-to-noise ratior (SNR), where the necessary singular value shrinkage is substantial.

\subsection{Image processing and denoising}

There are many approaches to denoising in image and signal processing, the majority designed for Gaussian noise \cite[see e.g.,][]{starck2010sparse}.  Most classical methods are designed for ``single-image denoising'', and do not share information across multiple images.  Our setting is different, because we have many very noisy samples---e.g., XFEL images. %Instead, these prior works preprocess raw images by extracting small patches from a single image, and consider those patches as their data vectors. In our experiments, we do not use such preprocessing, because we also want to estimate covariance of the image space, rather than the patch space, especially in the context of 3D structure reconstruction from XFEL patterns. %Some of the key approaches exploit sparsity, either in a fixed basis or dictionary---such as Fourier or wavelet---or in a basis that is estimated (or learned) from the data. 

\cite{starck2010sparse} Sec. 6.5. provides an overview of the classical methods for Poisson noise. Popular approaches reduce to the Gaussian case by a wavelet transform such as a Haar transform \citep{nowak1999wavelet}; by adaptive wavelet shrinkage; or by approximate variance stabilization such as the Anscombe transform. The latter is known to work well for Poisson signals with large parameters, due to approximate normality. However, the normal approximation breaks down for the Poisson with a small parameter, such as photon-limited XFEL \cite[see e.g.,][Sec. 6.6]{starck2010sparse}.

Other methods are based on singular value thresholding (SVT), with various approaches to handling non-Gaussian noise. For example, \citet{Furnival2016} performs SVT of the data matrix of image time-series in low noise, picking the regularization parameter to minimize the Poisson-Gaussian Unbiased Risk Estimator. We instead homogenize the data and propose a second-moment based denoising method.  Alternatively, \citet{Cao2014} frames denoising as a regularized maximum likelihood problem and uses SVT to optimize an approximation of the Poisson likelihood. Our approach avoids nonconvex likelihood optimization problems. %Another approach applies the alternating minimization-based PCA extension proposed in \citet{Collins2001} with refinements to exploit self-similarity in natural images \citep{salmon2014poisson}.  

%There are a variety of additional methods, including Bayesian ones. For instance, \citet{Sonnleitner2016} proposed to compute the Bayesian posterior distribution of the intensity at each pixel using local photon counts and carefully chosen priors, but this local smoothing method does not scale to photon-limited settings where the average photon count is much less than one.

\section{Covariance estimation}
\label{covar_norotat} 

$e$PCA is the eigendecomposition of a new covariance matrix estimator. To develop this estimator, we start with the sample covariance matrix and propose a sequence of improvements (see Table \ref{cov_tab} and below).

\begin{table}[h] 
\footnotesize
\centering
\caption{Covariance estimators}

\renewcommand{\arraystretch}{1.4}
\begin{tabular}{|l|l|l|l|l|l} 
\cline{1-5}
Notation& Name&Formula& Defined in & Motivation &  \\ \cline{1-5}
$S$&Sample covariance&$S =n^{-1}\sum_{i=1}^{n} (Y_i-\bar Y) (Y_i-\bar Y)^\top$ & \eqref{sam_cov} & -&  \\
$S_d$& Diagonal debiasing& $S_d = S - \diag[V(\bar Y)]$ &\eqref{diag_deb_est}& Hierarchy &  \\
$S_{h}$  & Homogenization& $S_{h} = D_n^{-\sfrac{1}{2}}S_d D_n^{-\sfrac{1}{2}}$ & \eqref{whit_est}    & Heteroskedasticity & \\
$S_{h,\eta}$& Shrinkage& $S_{h,\eta} = \eta(S_{h})$&\eqref{shr_est}&High dimensionality&  \\ 
$S_{he}$& Heterogenization& $S_{he} = D_n^{\sfrac{1}{2}}S_{h,\eta}D_n^{\sfrac{1}{2}}$&\eqref{rec_est}&Heteroskedasticity&  \\  
$S_{s}$& Scaling & $S_{s}=\sum \hat \alpha_i \hat v_i \hat v_i^\top$, where $S_{he} = \sum \hat v_i \hat v_i^\top$ &\eqref{scaled_cov}&Heteroskedasticity&  \\ 
\cline{1-5}
\end{tabular}
\label{cov_tab}
\end{table}

\begin{algorithm}[t]
	\SetAlgoLined
	\KwIn{Data $Y =[Y_1,\cdots,Y_n]^\top \in \R^{n\times p}$; Desired rank $r \le p$; \\ Mean-variance map $V$ of exponential family, as defined in (\ref{eqn:V}).}
	\KwOut{Covariance estimator $S_s \in \R^{p\times p}$ of noiseless vectors;
		$e$PCA: eigendecomposition of $S_s$.}
	
	Compute the sample mean $\bar Y = n^{-1}\sum_{i=1}^n Y_i$\\
	Compute the sample covariance matrix $S = n^{-1}\sum_{i=1}^n (Y_i-\bar Y)(Y_i-\bar Y)^\top$\\
	Compute the variance estimates $D_n  = \diag[V(\bar Y)]$\\
	Homogenize and diagonally debias the covariance matrix $\smash{S_{h}  = D_n^{-\sfrac{1}{2}} S D_n^{-\sfrac{1}{2}}} - I_p$ \\
	Compute the eigendecomposition $\smash{S_{h} = \hat W \Lambda \hat W ^\top}$\\
	Shrink the eigenvalues $S_{h,\eta} = \hat W \eta(\Lambda_r)\hat W^\top = \sum_{i=1}^r \hat \ell_i \hat w_i \hat w_i^\top$ of top $r$ eigenvalues $\Lambda_r=\diag(\lambda_1,\ldots,\lambda_r)$. %Diagonal debiasing can be performed in the same step by subtracting 1 from each eigenvalue.
	\\
	Compute the scaling coefficients $\hat \alpha_i = [1 - s^2(\hat \ell _i; \gamma) \tau_i]/  c^2(\hat \ell _i; \gamma)$ (as in \eqref{alpha_def})\\
	Heterogenize the covariance matrix $S_{he}= D_n^{\sfrac{1}{2}}S_{h,\eta} D_n^{\sfrac{1}{2}}$\\
	Scale the covariance matrix $S_s=\sum \hat \alpha_i \hat v_i \hat v_i^\top$, where the eigendecomposition of $S_{he}$ is $S_{he}=\sum \hat v_i \hat v_i^\top$\\
	\caption{Covariance matrix estimation and $e$PCA}\label{cov_est_alg}
\end{algorithm}

We will work with observations $Y$ from the canonical one-parameter exponential family with density 

\beq
\label{ef_dist}
p_\theta(y)  = \exp[\theta y - A(\theta)]
\eeq

\noindent with respect to a $\sigma$-finite measure $\nu$ on $\mathbb{R}$ (see e.g., \cite{lehmann2005testing}).
 Here $\theta \in \R$ is the natural parameter of the family and $A(\theta)  = \log \int  \exp(\theta y) d\nu(y) $ is the log-partition function. We assume the distribution is well-defined for all $\theta$ in an open set. 
 The mean and variance of $Y$ can be expressed as $\E Y = A'(\theta)$ and $\Va[Y] = A''(\theta)$, where we denote $g'(\theta)= dg(\theta)/d\theta$. 
 
Our running example will be the Poisson distribution $y\sim \Po(x)$. Here the carrier measure is the discrete measure with density $\nu(dy) = 1/y!$ with respect to the counting measure on the non-negative integers, while $\theta = \log(x)$ and $A(\theta) = \exp(\theta)$.

\subsection{The observation model}
\label{obs_mod}

Let $Y\in \R^p$ be a random vector with some unknown distribution. We observe $n$ i.i.d.~noisy data vectors $Y_i \sim Y$. In the XFEL application, $Y$ is the noisy image with the pixels as coordinates. We consider the following hierarchical model for $Y$. First, a latent vector---or hyperparameter---$\theta \in \R^p$ is drawn from a probability distribution $D$ with mean $\mu_{\theta}$ and covariance matrix $\Sigma_{\theta}$. Conditional on $\theta$, the coordinates of $Y = (Y(1),\ldots,Y(p))^\top$ are drawn independently from an exponential family $Y(j) \sim p_{\theta(j)}(y)$ defined in \eqref{ef_dist}. Formally, denoting by $\dot{\sim}$ the mean and the covariance of a random vector: 
\begin{align*}
 \theta \, & \dot{\sim}\,  (\mu_{\theta}, \Sigma_{\theta}) \\
Y(j)|\theta(j)  & \sim p_{\theta(j)}(y),  \,\,\,\, Y = (Y(1),\ldots,Y(p))^\top. 
\end{align*}

 Therefore, the mean of $Y$ conditional on $\theta$ is 
$$X := \E(Y|\theta) = (A'(\theta(1)),\ldots,A'(\theta(p)))^\top = A'(\theta),$$ 
so the noisy data vector $Y$ can be expressed as $Y= A'(\theta)+ \tilde\ep$, with $\E(\tilde\ep|\theta)=0$, while the marginal mean of $Y$ is $\E Y =\E A'(\theta)$. Thus one can think of $Y$ as a noisy realization of the clean vector $X=A'(\theta)$. However, the latent vector $\theta$ is also random and varies from sample to sample. In the XFEL application, $A'(\theta)$ are the unobserved noiseless images. %In the RNA-seq applications, $A'(\theta)$ are the mean read counts of all genes, for a given cell. 

The assumption of conditional independence given $\theta$ may sound restrictive.  It means that all latent effects that induce correlations do so through $\theta$ and not through some other mechanism. However, we can always capture some of the latent correlations in the ``mean'' structure by increasing the number of PCs. In addition, similar conditional independence is also common in empirical work such as  bulk RNA-Seq analysis \citep[e.g.,][]{anders2010differential}.

It is important that we model the \emph{mean} $A'(\theta)$ of the exponential family as our clean signal, as opposed to the \emph{natural parameter} $\theta$. One reason is that this enables a simple and deterministic algorithm, in contrast to the typical likelihood methods. Another reason is that in many applications, it is reasonable to assume that the means of noisy signals have a ``low-complexity'' structure, such as lying on a low-dimensional linear subspace. For instance, \cite{Basri2003} found that the images of a single face under different lighting conditions inhabit an approximately 9-dimensional linear space. As mentioned in Sec. \ref{rel_work}, this is a key modelling assumption distinguishing our approach from prior work like \cite{Collins2001}.

We thus have $Y = A'(\theta) + \diag[A''(\theta)]^{\sfrac{1}{2}}\ep$, where the coordinates of $\ep$ are conditionally independent and standardized given $\theta$. Therefore, the covariance of $Y$ conditional on $\theta$ is 
$$\C[Y|\theta]=\diag[A''(\theta(1)),\ldots,A''(\theta(p))] = \diag[A''(\theta)].$$  
 The marginal covariance of $Y$ is given by the law of total covariance: 
\beq
\label{cov}
\C[Y]  = \C[\E(Y|\theta)]  + \E [\C[Y|\theta]] =\C[A'(\theta)] + \E \diag[A''(\theta)].
\eeq

For Poisson observations $Y\sim \Po_p(X)$, where $X \in \R^p$ is random, we can write $Y = X + \diag(X)^{\sfrac{1}{2}}\ep$. The natural parameter is the vector $\theta$ with $\theta(j) = \log X(j)$. Since $A'(\theta(j))  = A''(\theta(j))= \exp(\theta(j)) = X(j)$, we see $\E Y = \E X $, and 
$\C[Y]  =  \C[X]+\E \diag[X].$

\subsection{Diagonal debiasing}
\label{diag_deb}

We will propose several estimators of increasing sophistication to estimate the covariance matrix $\Sigma_x =  \C[A'(\theta)]$ of the noiseless vectors $X_i=A'(\theta_i)$ (see Table \ref{cov_tab}). Clearly, due to the covariance equation \eqref{cov}, the sample covariance matrix of $Y_i$ is biased for estimating the diagonal elements of $\Sigma_x$. Fortunately, this bias can be corrected. Indeed, we only need to subtract the noise variances $\E A''(\theta(j))$. We know that $\E Y(j) = \E A'(\theta(j))$, so it is natural to define associated estimators via the \emph{variance map} of the exponential family, which takes a mean parameter $A'(\theta)$ into the associated variance parameter $A''(\theta)$. Formally, 
\begin{equation} \label{eqn:V}
V(m) = A''[(A')^{-1}(m)].
\end{equation}
If the distribution of $Y$ is non-degenerate, $A''(\theta) = \Va_{\theta}(Y)>0$, so $A'$ is increasing and invertible, and the variance map is well-defined. 

We define the sample covariance estimator 
\beq
\label{sam_cov}
S = n^{-1}\sum_{i=1}^{n} (Y_i -\bar Y) (Y_i -\bar Y)^\top,
\eeq
where $\bar Y = n^{-1} \sum_{i=1}^{n}Y_i $ is the sample mean. We estimate $\E A''(\theta)$ by $V(\bar Y)$, and define the \emph{diagonally debiased} covariance estimator 
\beq
\label{diag_deb_est}
S_d = S - \diag[V(\bar Y)].
\eeq

Continuing with our Poisson example, $A'(\theta)=A''(\theta) = \exp(\theta)$, so $V(m) = m$, and $S_d = S - \diag[\bar Y]$. In this example the estimator is unbiased, because $V$ is linear. When $V$ is non-linear, the estimator can become slightly biased.  

\subsubsection{The rate of convergence}
\label{sec_rate_conv}
Our first theoretical result characterizes the finite-sample convergence rate of the diagonally debiased covariance estimator $S_d$, for any fixed $n,p$. This estimator is not a sample covariance matrix, which is inconsistent in our case when $n\to\infty$ and $p$ is fixed. Thus it is necessary to study its convergence rate from first principles.

For this we need to make a few technical assumptions. First, we assume that the variance map $V$ is Lipschitz with constant $L$. It is easy to check that this is true for the Gaussian and Poisson distributions. We also assume that the coordinates of the random vector $\theta$ are almost surely bounded, $\|\theta\|_{\infty} \le B$. Since $A'$ is continuous and invertible, this is equivalent to the boundedness of $A'(\theta)$. This is reasonable in the areas that we are interested in---XFEL imaging does not have infinite energy, so we have an upper bound on the intensity of pixels. Finally we assume that $m_4 = \max_i \E[Y(i)^4] \ge C$ for some universal constant $C>0$. This is reasonable, as it states that at least some entries of the random vector have non-vanishing magnitude. 

Let $\lesssim$ denote inequality up to constants not depending on $n$ and $p$. Let $\V \cdot \V_{\Fr}$ be the Frobenius norm and $\V \cdot \V$ be the operator norm. Our result, proved in Sec. \ref{pf_rate_low_dim}, is

\begin{theorem}[Rate of convergence of debiased covariance estimator] The diagonally debiased covariance estimator $S_d $ has the following rates of convergence. In Frobenius norm, with $\mu:=\E Y = \E X = \E A'(\theta)$: 
	$$\E[\V S_d - \sig_x\V_{\Fr}] \lesssim \sqrt{\frac{p}{n}} 
	\left[\sqrt{p} \cdot m_4 
+\| \mu \| \right].$$
In operator norm, with the dimensional constant $C(p) =4(1+2 \lceil \log p\rceil)$: 
	$$\E[\V S_d - \sig_x\V] \lesssim  \sqrt{C(p)}  \frac{(\E\|Y\|^4)^{\sfrac{1}{2}} + (\log n)^3(\log p)^2}{\sqrt{n}}
	+ \sqrt{\frac{p}{n}} \left[
1 + \sqrt{\frac{p}{n}}
+\| \mu \| \right].$$
	\label{rate_low_dim}
\end{theorem}

The two error rates are both of interest, and complement each other. % Neither error rate implies the other. The proof for the Frobenius norm rate only uses 4-th moment bounds for $Y(i)$, while the operator norm rate uses higher moments---of order $\log(p)$---which relies on more of the strength of the exponential family.  
The Frobenius norm rate captures the deviation across all entries of the covariance matrix.
%Since the squared Frobenius norm is approximately a sum of $p^2$ squared deviations of sample means from population means, each of which is $O(n^{-1})$, we expect a rate of $\smash{\sfrac{p}{n^{\sfrac{1}{2}}}}$ in Frobenius norm. If the bound $m_4$ on the fourth moment is of order one, that is exactly what we obtain. It is not completely clear that our rates are optimal, but this suggests so for the Frobenius norm.
%In contrast the error rate in operator norm captures the deviation in the extreme eigenvalues of the error matrix. 
The operator norm rate is typically faster than the Frobenius norm rate. For instance, in XFEL it is reasonable to assume that that the total intensity across all detectors is fixed as the resolution increases.  This leads to a fixed value for $\E\|Y\|^4$ that does not grow with $n$. The operator norm rate can be as fast as $(p/n)^{\sfrac{1}{2}}$ while the Frobenius norm rate is $p/n^{\sfrac{1}{2}}$.

Our proof of Thm. \ref{rate_low_dim} exploits that exponential family random variables are sub-exponential, so we can use corresponding moment bounds. We also rely on operator-norm bounds for random matrices from \cite{Tropp2015} and on moment bounds from \cite{Boucheron2005}.

\section{Homogenization and shrinkage}
\label{cov_est}

\subsection{Homogenization}
\label{sec_white}

In the previous sections, we showed that the diagonally debiased sample covariance matrix converges at a rate $O(pn^{-\sfrac{1}{2}})$. Next we propose a shrinkage method to improve this estimator in the high dimensional regime where $n, p\to \infty$ and $p/n\to\gamma>0$.  As a preliminary step, it is helpful to homogenize the empirical covariance matrix and remove the effects of heteroskedasticity. This allows us to get closer to the \emph{standard spiked model} \citep{johnstone2001distribution} where the noise has the same variance for all features. In that setting covariance estimation via eigenvalue shrinkage has been thoroughly studied \citep{Donoho2013}. 

The vector of noise variances affecting the different components is $\E[A''(\theta)]$. For a given signal $Y = A'(\theta) + \diag[A''(\theta)]^{\sfrac{1}{2}}\ep$, homogenization transforms it to $Y_h = \diag[A''(\theta)]^{-\sfrac{1}{2}}A'(\theta) + \ep$. The covariance is transformed from $\Cov{Y}$ to $\diag[A''(\theta)]^{-\sfrac{1}{2}}\Cov{Y} \diag[A''(\theta)]^{-\sfrac{1}{2}}$. Since the diagonal correction  $D_n  = \diag[V(\bar Y)]$  estimates $\E \diag[A''(\theta)]$, we define the \emph{homogenized} covariance estimator by
\beq
\label{whit_est}
S_{h}=D_n^{-\sfrac{1}{2}}S_d D_n^{-\sfrac{1}{2}} =D_n^{-\sfrac{1}{2}} S D_n^{-\sfrac{1}{2}} - I_p.
\eeq
For Poisson observations, every entry of the noisy vector has to be divided by square root of the corresponding entry of the sample mean, so  $S_{h}  = \diag[\bar Y]^{-\sfrac{1}{2}}S\diag[\bar Y]^{-\sfrac{1}{2}} - I_p$. 

Homogenization is different from \emph{standardization}, the classical method for removing heteroskedasticity. To standardize, each feature---e.g., pixel---is divided by its empirical standard deviation \cite[e.g.,][Sec. 2.3.]{jolliffe2002principal}. This ensures that all features have the same norm. The sample covariance matrix becomes a sample correlation matrix. In our case it turns out that this procedure ``over-corrects''. The overall variance $\Va[Y(i)]$ of each feature is the sum of the signal variance $\Va[A'(\theta(i))]$ and the noise variance $\E[A''(\theta(i))]$. Homogenization divides by the estimated noise standard errors, while standardization divides by the \emph{overall} standard error due to the signal and noise. 

Therefore, in our setting homogenization is more justified than standardization. Moreover, the standard Marchenko-Pastur law holds for the homogenized estimator (Thm. \ref{MP_white} in the next section). This also suggests that the top ``noise'' eigenvalue has a well-understood Tracy-Widom distribution asymptotically \citep{johnstone2001distribution}, which can be used to devise tests of significance.  Another justification is that standardization improves the signal strength for ``delocalized'' eigenvectors (Sec. \ref{SNR}). We discuss these in detail below.

\subsubsection{Marchenko-Pastur law}
\label{mp_law}
A key advantage of homogenization is that the homogenized estimator has a simple well-understood asymptotic behavior. In contrast, the unhomogenized estimator has a more complicated behavior. In this section, we show both of the above claims. We show that the limit spectra of our covariance matrix estimators are characterized by the Marchenko-Pastur (MP) law \citep{marchenko1967distribution}, proving the general MP law for the sample covariance $S$, and the standard MP law for the homogenized covariance $S_{h}$. 

For simplicity, we consider the case is when $\theta \in \mathbb{R}^p$ is fixed. This can be thought of as the ``null'' case, where all mean signals are the same. Then we can write $Y_i = A'(\theta) + \diag[A''(\theta)]^{\sfrac{1}{2}}\ep_i$, where $\ep_i$ have independent standardized entries. Therefore, letting $\mathcal{Y}$ be the $n\times p$ matrix whose rows are $Y_i^\top$, we have  
$\mathcal{Y} = \vec{1} A'(\theta)^\top + \mathcal{E} \diag[A''(\theta)]^{\sfrac{1}{2}}$, where $ \vec{1}  = (1,1,\ldots,1)^\top$ is the vector of all ones, and $\mathcal{E}$ is an $n \times p$ matrix of independent standardized random variables.

Let $H_p$ be the uniform distribution on the $p$ scalars $A''(\theta(i))$, $i=1,\ldots,p$.  We assume that $A''(\theta(i))>c$ for some universal constant $c>0$. In the Poisson example, this means that the individual rates $x(i)$ are bounded away from 0. The reason for this assumption is to avoid the very sparse regime, where only a few nonzero entries per row are observed. In that case, the MP law is not expected to hold. 

Consider the high dimensional asymptotic limit when $n,p\to\infty$ so that $p/n\to\gamma>0$. Suppose moreover that $H_p$ converges weakly to some limit distribution, i.e., $H_p \Rightarrow H$. Since $\diag[A''(\theta)]$ can be viewed as the population covariance matrix  of the noise, $H$ is the limit population spectral distribution (PSD). Since $\mathcal{E}$ has independent standardized entries with bounded moments, it follows that the distribution of the $p$ eigenvalues of $n^{-1}\mathcal{Y}^\top \mathcal{Y}$ converges almost surely to the general Marchenko-Pastur distribution $F_{\gamma,H}$ \citep[][Thm. 4.3]{bai2009spectral}. 

Now, the sample covariance matrix $S$ is a rank-one perturbation of $n^{-1}\mathcal{Y}^\top \mathcal{Y}$. Therefore its eigenvalue distribution also converges to the MP law. We state this for comparison with the next result. 

\begin{prop}[Marchenko-Pastur law for sample covariance matrix]
\label{MP_unhomogenized}
The eigenvalue distribution of $S$ converges almost surely to the general Marchenko-Pastur distribution $F_{\gamma,H}$. 
\end{prop}

Since the general MP law has a complicated implicit description that needs to be studied numerically \citep[see e.g.,][]{dobriban2015efficient}, it is useful to work with the homogenized covariance matrix  $S_{h}$. Indeed, we establish that the standard Marchenko-Pastur law characterizes its limit spectrum.  The standard Marchenko-Pastur distribution has a simple closed-form density, and there are many useful tools already available for low-rank covariance estimation \citep[e.g.,][]{shabalin2013reconstruction, Donoho2013}. 

\begin{theorem}[Marchenko-Pastur law for homogenized covariance matrix]
\label{MP_white}
The eigenvalue distribution of $S_{h}+I_p$ converges almost surely to the standard Marchenko-Pastur distribution with aspect ratio $\gamma$. 
\end{theorem}

In the proof presented in Appendix \ref{pf_MP_white}, we deduce this from the Marchenko-Pastur law for the error matrix $n^{-\sfrac{1}{2}}\mathcal{E}$, for which standard results from \cite{bai2009spectral} apply. The emergence of the standard MP law motivates the shrinkage method presented next.

\subsection{Eigenvalue shrinkage} 
\label{eval_shr}
Since the early work of Stein \citep{stein1956some} it is known that the estimation error of the sample covariance can be decreased by eigenvalue shrinkage. Therefore, we will apply an eigenvalue shrinkage method to the homogenized covariance matrix $S_{h}$. Let $\eta(\cdot)$ be a generic matrix shrinker, defined for symmetric matrices $M$ with eigendecomposition $M = U \Lambda U^\top$ as $\eta (M) = U \eta(\Lambda) U^\top$. Here $\eta(\Lambda)$ is defined by applying the scalar shrinker $\eta$---typically a nonlinear function---elementwise on the diagonal of the diagonal matrix $\Lambda$. Then our \emph{homogenized and shrunken} estimators will have the form 
\beq
\label{shr_est}
S_{h,\eta}  = \eta(S_{h})= \eta(D_n^{-\sfrac{1}{2}}S_d D_n^{-\sfrac{1}{2}}).
\eeq
We are interested in settings where the clean signals lie on a low-dimensional subspace. We then expect the true covariance matrix $\Sigma_x$ of the clean signals to be of low rank. However, based on Thm. \ref{MP_white}, even in the case when $\Sigma_x = 0$, the empirical homogenized covariance matrix is of full rank, and its eigenvalues have an asymptotic MP distribution. We are thus interested in shrinkers $\eta$ that set all noise eigenvalues to zero, specifically $\eta(x)=0$ for $x$ within the support of the shifted MP distribution $x \in [(1-\sqrt{\gamma})^2, (1+\sqrt{\gamma})^2]-1$. An example is operator norm shrinkage \citep{Donoho2013}. %Alternatively, in many exploratory analyses, one may start with a guess for the desired number $r$ of principal components. In that case, one can keep only the top $r$ components after homogenization. 

\begin{figure}[t]
	\centering
	\begin{subfigure}{.5\textwidth}
		\centering
		\includegraphics[width=1.05\textwidth, trim= 50 20 10 0,clip]{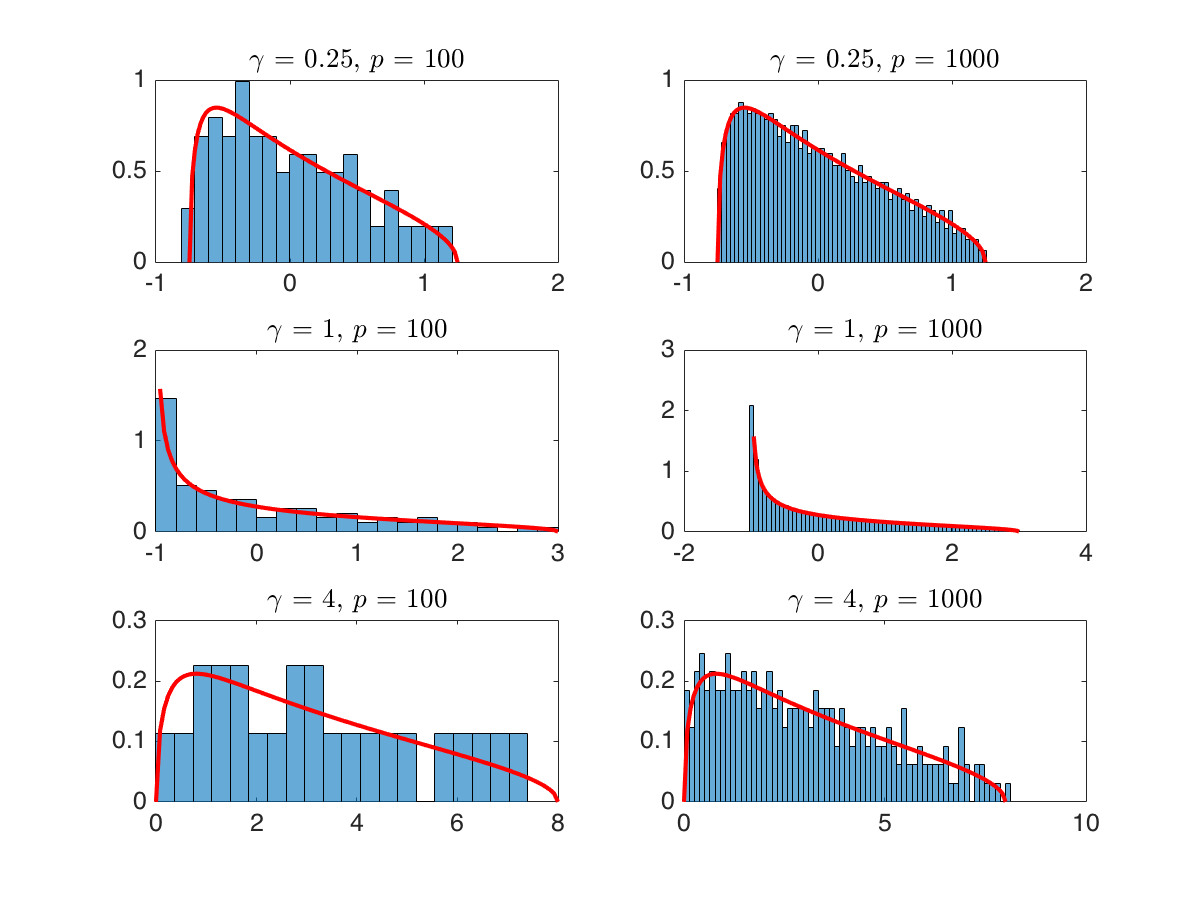} \caption{Rank 0 case} \label{whitened_estim_MPoverlay}
	\end{subfigure}%
	\begin{subfigure}{.5\textwidth}
		\centering
		\includegraphics[width=1.05\textwidth, trim= 50 20 10 0 ,clip]{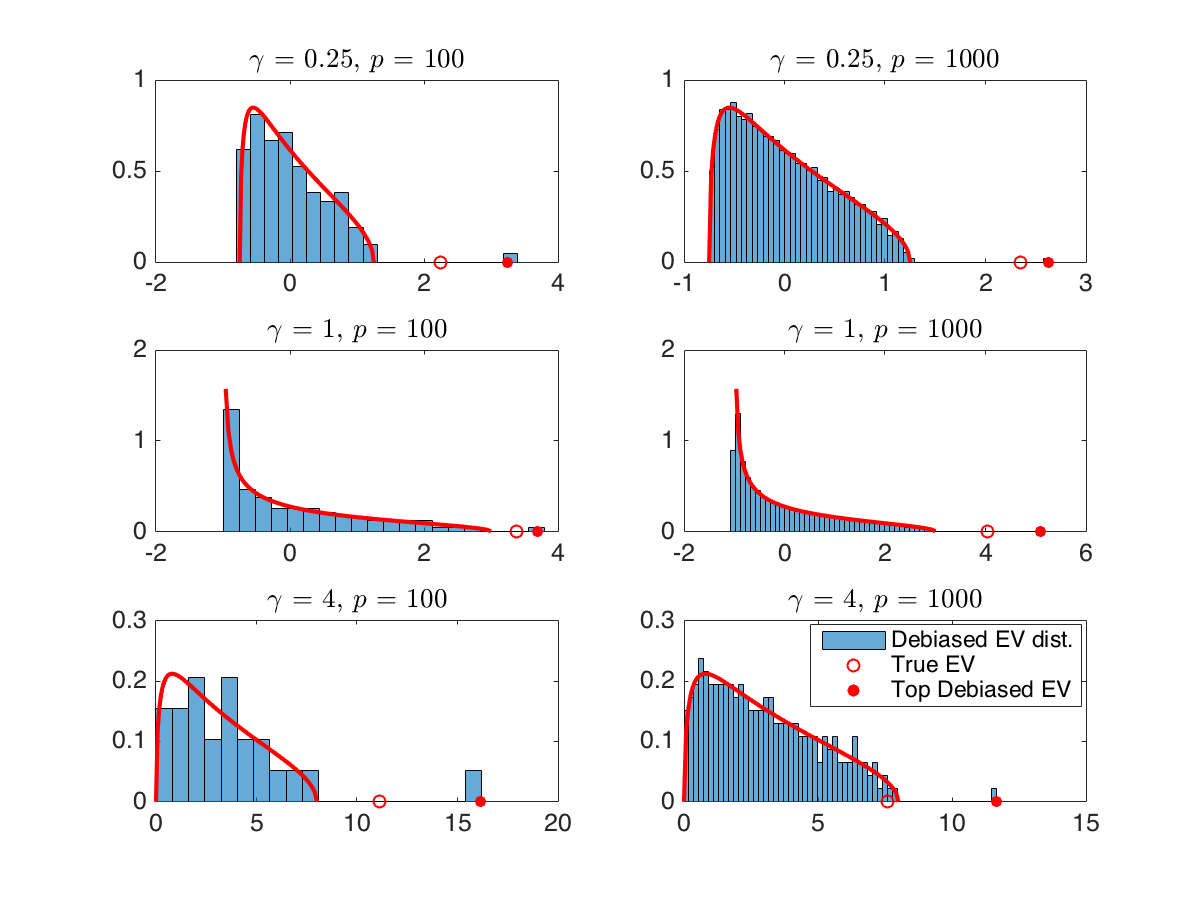} \caption{Rank 1 case (Spiked Model)} \label{spike_whitened_estim}
	\end{subfigure}
	\caption{Empirical distribution of eigenvalues of homogenized sample covariance $S_{h}$ for different values of $\gamma = p/n$, with the corresponding shifted Marchenko-Pastur density overlaid as a red curve. Data simulated according to \ref{eval_shr}. In the legend for (b), `Top Debiased EV' refers top eigenvalue of $S_{h}$, while `True EV' refers to the top eigenvalue of $D_n^{-\sfrac{1}{2}}\Sigma_xD_n^{-\sfrac{1}{2}}$, which we want to estimate.}
	\label{spectrum_homogenized}
\end{figure}

However, homogenization by $D_n\neq I_p$ also changes the direction of the eigenvectors. Therefore, to improve the accuracy of subspace estimates after eigenvalue shrinkage, we \emph{heterogenize}, multiplying back by the estimated standard errors. We define the \emph{heterogenized} covariance estimator as:
\beq
\label{rec_est}
S_{he}  = D_n^{\sfrac{1}{2}}\cdot S_{h,\eta}\cdot D_n^{\sfrac{1}{2}}.
\eeq

Heterogenization is a non-linear operation that changes both the eigenvectors and eigenvalues. While it improves the estimates of the eigenvectors (PCs), it turns out that it introduces a bias in the eigenvalues. Therefore, we will need a final \emph{scaling} step to correct this bias (Sec. \ref{bc_ev}). 

To understand homogenization empirically, we perform two simulations. First, we generate non-negative i.i.d $\{X_i\}_{1\leq i \leq n}$ lying in a low-dimensional space of dimension $r$: we pick $r$ vectors $v_1, …, v_r \in \R^p$ whose coordinates are i.i.d uniformly distributed in $[0,1]$, and normalize each to have an L1 norm of unity. For each $i$, sample $r$ coefficients $a_{i1}, …, a_{ir}$ independently from the uniform distribution on [0,1]. Define
$X_i = a_{i1}v_1 + … + a_{ir} v_r$. Note that $X_i$ are non-negative, reside in a hyperplane spanned by $v_1,…,v_r$, and the mean and covariance of $X_i$ can be found easily in terms of $v_1,…,v_r$. The coefficients $a_{i1},…,a_{ir}$ are also normalized so that $a_{i1}+…+a_{ir}=A$, where $A=25(1+\sqrt{\gamma})^2$ is a constant relating to signal strength, chosen empirically to push the top eigenvalue outside of the bulk.  Finally we sample $Y_i \sim \Po_p(X_i)$ independently.

We display a Monte Carlo instance of the eigenvalue histogram of $S_{h}$ on Figure \ref{spectrum_homogenized}. When $r=0$, the standard MP distribution---shifted by $-1$---is a good match (Fig. \ref{whitened_estim_MPoverlay}). This is in accordance with Thm. \ref{MP_white}. When $r=1$, the standard MP distribution still matches the bulk of the noise eigenvalues (Fig. \ref{spike_whitened_estim}). Moreover, we observe the same qualitative behaviour as in the classical spiked model, where the top \emph{empirical eigenvalue} overshoots the \emph{population eigenvalue}. Next we study this phenomenon more precisely.

\subsubsection{The spiked model: Colored and homogenized}
\label{spiked_sec}

\begin{table}[]
\centering
\caption{Spiked models: Summary of the original and homogenized spiked model.}
\label{spiked_tab}
{\renewcommand{\arraystretch}{1.4}
\begin{tabular}{|l|l|l|}
\hline
Model  & Original & Homogenized \\ \hline
Latent Signal  & $X_i = u + z_i v$ & $D^{-\sfrac{1}{2}} X_i = D^{-\sfrac{1}{2}} u + z_i D^{-\sfrac{1}{2}}v$  \\ \hline
Marginal Covariance & $\C[Y]= v v^\top+D$ & $\C[Y_h]= D^{-\sfrac{1}{2}} v v^\top D^{-\sfrac{1}{2}}+ I_p$ \\ \hline
Eigenvector  & $v_{norm}= v/\|v\|$  & $w = D^{-\sfrac{1}{2}} v / \|D^{-\sfrac{1}{2}} v\|$\\ \hline
Spike   & $t = v^\top v$ & $\ell =v^\top D^{-1} v $  \\ \hline
SNR   & $\frac{v^\top v}{\tr{D}}$ & $\frac{v^\top D^{-1} v}{p}$ \\ \hline
\end{tabular}
}
\end{table}

To develop a method for estimating the eigenvalue after homogenization and heterogenization, we study a generalization of the spiked model \citep{johnstone2001distribution} appropriate for our setting. Specifically, based on the covariance structure of the noisy signal, Eq. \eqref{cov}, we model the mean parameter $X = A'(\theta)$ of the exponential family---the clean observation---as a low rank vector. 
For simplicity, we will present the results in the rank one case, but they generalize directly to higher rank.

Suppose that the $i$-th clean observation has the form $X_i  = A'(\theta_i) = u + z_i v$, where $u,v$ are deterministic $p$-dimensional vectors, and $z_i$ are i.i.d.~standardized random variables. In the Poisson case where $Y_i \sim \Po_p(X_i)$, this assumes that the latent mean vectors are $X_i = u + z_i v$. The vector $u$ is the global mean of the clean images, while  $v$ denotes the direction in which they vary.

For  $X_i$  to be a valid mean parameter, we need the additional condition that $u(j) +z_i |v(j)| \in A'(\Theta)$, for all $i,j$, where $\Theta$ is the natural parameter space of the exponential family, and $f(S)$ denotes the forward map of the set $S$ under the function $f$. For instance, in the Poisson case, we need that $X_i(j) \ge 0$ for all $i,j$. If we take $z_i$ to be uniform random variables on $[-\sqrt{3},\sqrt{3}]$, so that their variance is unity, then a sufficient condition is that $ u(j) \ge \sqrt{3}|v(j)|$ for all $j$.

Using our formula for the marginal covariance of the noisy observations, $\C[Y]  =  \C[X]+\E \diag[V(X)]$, and defining  $D = \E \diag[V(X)]$, we obtain
\beq\label{spiked}
\C[Y]  = v v^\top+D.
\eeq
For instance, in the Poisson case we have $\C[Y]  = v v^\top+\diag[u]$.

We homogenize the observations dividing by the elements of $D^{\sfrac{1}{2}}$. The elements of $D$ are expected values of variances. They are thus positive, except for coordinates that that can be discarded because they have no variability.  The homogenized observations are $Y_h = D^{-\sfrac{1}{2}} Y$, and their population covariance matrix is
\beq\label{spiked_w}
\C[Y_h]  = D^{-\sfrac{1}{2}} v v^\top D^{-\sfrac{1}{2}}+ I_p.
\eeq

We now compare this with the usual \emph{standard spiked model} \citep{johnstone2001distribution} where the observations $Y_h$ are Gaussian and have covariance matrix
$\C[Y_h]  = \ell w w ^\top + I_p,$
where $\ell\ge 0$ and the vector $w$ has unit norm. The top eigenvalue is called the ``spike''. This model has been thoroughly studied in probability theory and statistics. In particular, the Baik-Ben Arous-P\'ech\'e (BBP) phase transition (PT) \citep{baik2005phase} shows that when $n,p\to\infty$ such that $p/n \to \gamma>0$, the top eigenvalue of the sample covariance matrix asymptotically separates from the Marchenko-Pastur bulk if the population spike $\ell> \sqrt{\gamma}$. Otherwise, the top sample eigenvalue does not separate from the MP bulk. This was shown first for complex Gaussian observations, then generalized to other distributions \citep[see e.g.,][]{yao2015large}.

Heuristically, comparing with \eqref{spiked_w}, we surmise that a spiked model with $\ell =  v^\top D^{-1} v $ and $w = D^{-\sfrac{1}{2}}v/$  $\|D^{-\sfrac{1}{2}} v\|$ is a good approximation in our case. In particular the BBP phase transition should happen approximately when
$ v^\top D^{-1} v  = \sqrt{\gamma}.$
In the Poisson case the condition is $v^\top \diag[u]^{-1} v  = \sqrt{\gamma}$. Next we provide numerical evidence for this surmise, and develop its consequences.

\subsubsection{Homogenization improves SNR}\label{SNR}

In this section we justify our homogenization method theoretically, showing that it can improve the signal-to-noise ratio. This was observed empirically in previous work on covariance estimation in a related setting, but a theoretical explanation is lacking \citep{bhamre2016denoising}. 

As usual, we define the SNR of a ``signal+noise'' vector observation $y = s+n$ as the ratio of the trace of the covariances of $s$ and of $n$. In the unhomogenized model from Eq. \eqref{spiked}
$$\text{SNR}=\frac{\tr{\Cov X}}{\tr{\E\diag[V(X)]}} =  \frac{\tr{vv^\top}}{\tr{D}} = \frac{v^\top v}{\tr{D}}.$$
In particular, the SNR is of order $O(1/p)$ in the typical case when the vector $v$ has norm of unit order. In the homogenized model from Eq. \eqref{spiked_w}, the SNR equals $v^\top D^{-1} v/p$. 

Suppose now that $v$ is approximately \emph{delocalized} in the sense that 
$p\cdot v^\top D^{-1} v \approx \tr D^{-1}\cdot v^\top v.$ 
This holds for instance if the entries of $v$ are i.i.d.~centered random variables with the same variance $\sigma^2$.  In that case, $\E v^\top D^{-1} v = \sigma^2 \tr D^{-1}$ and $\E v^\top  v = \sigma^2 p$, and under higher moment assumptions it is easy to show the concentration of these quantities around their means, showing delocalization as above. If $v$ is delocalized, then we obtain that the SNR in the homogenized model is higher than in the original model. Indeed, this follows because $D$ is diagonal, so by the Cauchy-Schwarz inequality
$$\frac{v^\top D^{-1} v}{p} 
\approx \frac{\tr D^{-1} \cdot v^\top v}{p^2}  
=  \frac{\sum_{i=1}^{p} D_i^{-1}  \cdot v^\top v}{p^2}   
\ge   \frac{v^\top v}{\sum_{i=1}^{p} D_i}  
= \frac{v^\top v}{\tr{D}}.$$
Moreover, we can define the \emph{improvement} (or \emph{amplification}) in SNR as
\beq\label{i_snr}
\I= \frac{\tr{D}}{p} \cdot \frac{v^\top D^{-1} v}{v^\top v}.
\eeq
The above heuristic can be formalized as follows:
\begin{prop}
\label{prop_snr_wh}
Suppose the signal eigenvector $v$ is delocalized in the sense that for some $\ep>0$, 
$$\frac{v^\top D^{-1} v}{v^\top v} \ge (1-\ep) \frac{\tr[D^{-1}]}{p}.$$
Let moreover $\beta$ be the following measure of heteroskedasticity: 
$$\beta =  \frac{\sum_{i=1}^{p} D_i \cdot \sum_{i=1}^{p} D_i^{-1}}{p^2} \ge 1.$$
Then the SNR is improved by homogenization, by a ratio $\I \ge (1-\ep)\beta$.
\end{prop}
If $\beta$ is large and $\ep>0$ is small, the SNR can improve substantially.

\subsubsection{Eigenvalue shrinkage and scaling} \label{bc_ev}

We now continue with our overall goal of estimating the covariance matrix $\C[X] = vv^\top$  of $X$.  This has one nonzero eigenvalue $t=\|v\|^2$ and corresponding eigenvector $v_{norm}  = v/\|v\|$. We use the top eigenvector of the heterogenized covariance matrix $S_{he}$ as an estimator of $v_{norm}$. To estimate $t$, a first thought is to use the top empirical eigenvalue of $S_{he}$, but as we show next, this naive estimator is biased. %To understand and correct the bias, we review some basic facts about the spiked model in high dimensions. 

For data with independent coordinates and equal variances, the cumulative work of many authors \cite[e.g.,][etc]{baik2005phase, baik2006eigenvalues,paul2007asymptotics, benaych2011eigenvalues} shows that if the population spike $\ell$ is above the BBP phase transition---i.e.,  $\ell > \sqrt{\gamma}$---then the top sample spike pops out from the Marchenko-Pastur distribution of the ``noise'' eigenvalues. The top eigenvalue will converge to the value given by \emph{the spike forward map}: 
\begin{equation*}
\lambda(\ell;\gamma)=
\left\{
	\begin{array}{ll}
		(1+\ell) \left(1+ \frac{\gamma}{\ell}\right) & \mbox{\, if \, } \ell>\gamma^{\sfrac{1}{2}}, \\
		(1+\gamma^{\sfrac{1}{2}})^2 & \mbox{\, otherwise.}
	\end{array}
\right.
\end{equation*}

We conjecture that the BBP phase transition also applies to our case, and describes the behavior of the spikes after homogenization. We have verified this in numerical simulations in certain cases (data not shown due to space limitations). Therefore, as in previous work, we propose to estimate $\ell$ consistently by inverting the spike forward map \cite[see e.g.,][]{lee2010convergence, Donoho2013}, i.e., defining $\smash{\hat \ell = \lambda^{-1}(\lambda_{\max}(S_{h}))}$. \cite{Donoho2013} provided an asymptotic optimality result for this estimator of the spike in operator norm loss. 

Once we have a good estimator $\hat \ell$ of $\ell  = v^\top D^{-1} v$, a first thought is to estimate $t = v^\top v$ as the top eigenvalue of the heterogenized covariance matrix $S_{he}$. However, this estimator is biased. The estimation accuracy is affected in a significant way by the inconsistency of the empirical eigenvector $\hat w$ of $S_{h}$ as an estimator of the true eigenvector $w = D^{-\sfrac{1}{2}} v / \|D^{-\sfrac{1}{2}} v\|$. We can quantify this heuristically based on results for Gaussian data. In the Gaussian standard spiked model the empirical and true eigenvectors have an asymptotically deterministic angle: $(w^\top \hat w)^2 \to c^2(\ell;\gamma)$ almost surely, where $c(\ell;\gamma)$ is the \emph{cosine forward map} given by \cite[e.g.,][etc]{paul2007asymptotics, benaych2011eigenvalues}:
\begin{equation*}
c(\ell;\gamma)^2=
\left\{
	\begin{array}{ll}
		\frac{1-\gamma/\ell^2}{1+\gamma/\ell} & \mbox{\, if \, } \ell>\gamma^{\sfrac{1}{2}}, \\
		0 &  \mbox{\, otherwise.}
	\end{array}
\right.
\end{equation*}
Heuristically, in finite samples we can write
$\hat w  \approx c w + s \ep,$
where $s=s(\ell;\gamma)\ge0$ is the sine defined by $s^2 = 1-c^2$, and $\ep$ is white noise with approximate norm $\|\ep\|=1$. Then, since  $w^\top D w  = v^\top v / v^\top D^{-1} v = t/\ell$, and $\ep^\top D \ep \approx \tr(D)/d$, we have 
$$
\|\hat v\|^2 
\approx \ell \cdot \hat w^\top D \hat w 
\approx \ell \cdot (c w + s \ep)^\top D (c w + s \ep)
\approx \ell \cdot (c^2 w^\top D w + s^2 \ep^\top D \ep) 
 \approx t c^2 + \ell s^2 \tr(D)/p.
$$
Comparing this to $\|v\|^2 = t = t c^2 + ts^2$, we find that the bias is 
$$\|\hat v\|^2 - t \approx s^2 \left(v^\top D^{-1}v \cdot \tr(D)/p - v^\top v\right) = s^2 t \cdot (\I-1) \ge 0.$$
This suggests that $\|\hat v\|^2$ is an upward biased estimator of $t =\|v\|^2$. Interestingly, the bias is closely related to the improvement $\I$ in SNR.

To correct the bias, we propose an estimator of the form $\hat t(\alpha) = \alpha \|\hat v\|^2$ for which $\alpha \|\hat v\|^2 \approx \|v\|^2$. We have $ \|\hat v\|^2 \approx t  \cdot [1+ s^2 (\I-1)]$, suggesting that we define $\alpha  = [1+ s^2 (\I-1)]^{-1}$. This quantity is an unknown population parameter, and it depends on $s^2$ and $\I$. We can estimate $s^2$ in the usual way by $\hat s^2 = s^2(\hat \ell;\gamma)$. Since $\I$ itself depends on the parameter $t$ we are trying to estimate, we plug in the same estimator $\hat t(\alpha) = \alpha \|\hat v\|^2$, leading to the following estimator of $\I$ (where we also define $\tau$ for future use):
$$\hat \I(\alpha) 
= \frac{\tr{D_n}}{p} \cdot \frac{\hat \ell}{\hat t(\alpha)} 
= \frac{\tr{D_n}}{p} \cdot \frac{\hat \ell}{\alpha \|\hat v\|^2} 
= \frac{\tau}{\alpha}.$$
Since  $\alpha  = [1+ s^2 (\I-1)]^{-1}$, it is reasonable to require that the fixed-point equation 
$\hat \alpha  = [1+ \hat s^2 (\hat \I(\hat \alpha)-1)]^{-1}$ holds.

We can equivalently rewrite the fixed-point equation as
$1/\hat \alpha = \hat c^2+ \hat s^2 \hat \I(\hat\alpha) = \hat c^2+ \hat s^2 \tau/\hat \alpha$. Or, when $\hat c^2>0$, 
\beq
\label{alpha_def}
\hat \alpha  = \frac{1-\hat s^2\tau}{\hat c^2}.
\eeq
When $\hat c^2=0$, i.e., when $\hat \ell \le \sqrt{\gamma}$, the equation reads $1/\hat \alpha  = \tau/\hat\alpha$. If $\tau=1$, this has solution $\alpha=1$, else it has no solution. Therefore, when $\hat c^2=0$, we define $\hat \alpha = 1$. We finally define $\hat t(\hat \alpha) = \hat\alpha \|\hat v\|^2$. The implication is that we ought to rescale the estimated magnitude of the signal subspace corresponding to $v$ by $\hat \alpha$. 

In the multispiked case, suppose $X_j = u + \sum_{i=1}^r z_{ij} v_i$. Then the marginal covariance of $Y$ is $\Cov{Y} = \sum_{i=1}^r v_i v_i^\top + D$. Suppose that the $v_i$ are sorted in the order of decreasing norm. Suppose moreover that the heterogenized sample covariance $S_{he}$ has the form
$S_{he} = \sum_{i=1}^{r} \hat v_i \hat v_i^\top =  \sum_{i=1}^{r} \hat\lambda_i \hat u_i \hat u_i^\top ,$
where $\hat u_i$ are orthonormal, and the $\hat\lambda_i \ge 0$ are sorted in decreasing order. Based on our above discussion, we define the \emph{scaled} covariance matrix as 
\beq
\label{scaled_cov}
S_s = \sum_{i=1}^{r} \hat \alpha_i \hat v_i \hat v_i^\top,
\eeq
where $\hat \alpha_i$ is defined in \eqref{alpha_def}, with $\hat s^2 = \hat s^2_i = s^2(\hat \ell_i;\gamma)$. This concludes our methodology for covariance estimation. We use the terminology $e$PCA for the eigendecomposition of the covariance matrix estimator \eqref{scaled_cov}. Both the eigenvalues and the eigenvectors of this estimator are different from those of the sample covariance matrix.

$e$PCA is summarized in Alg. \ref{cov_est_alg}. Clearly, $e$PCA is applicable when the variables $x(i)$ have known non-identical distributions, which the modification that homogenization should be done by the mean-variance map of the distribution of each particular coordinate. As discussed at the beginning of Sec. \ref{eval_shr}, we assume here that we have a guess $r$ for the number of PCs. In exploratory analyses, one can often try several choices for $r$. While there are many formal methods for choosing the rank $r$ \citep[see e.g.,][]{jolliffe2002principal}, it is beyond our scope to investigate them in detail here. %(see Sec. \ref{future_work}).

\subsubsection{Simulations with $e$PCA}

We report the results of a simulation study with $e$PCA.  We simulate data $Y_i$ from the Poisson model $Y_i \sim \Po_p(X_i)$, where the mean parameters are $X_i = u + z_i \ell^{\sfrac{1}{2}} v$, the $z_i$ are i.i.d.~unit variance random variables uniformly distributed on $[-\sqrt{3},\sqrt{3}]$, and $u \in \mathbb{R}^{p}$ has entries $u(i)$ sorted in increasing order on a uniform grid on $[1,3]$, while $v \in \mathbb{R}^{p}$ has entries $v(i)$ sorted in increasing order on a uniform grid on $[-1,1]$, standardized so that $\|v\|^2 = 1$. We take the dimension $p = 500$, and $\gamma = \sfrac{1}{2}$, so $n=1000$. The phase transition occurs when the spike is $\ell  = \sqrt{\gamma}/v^\top \diag[u]^{-1} v \approx 1.2$. We vary the spike strength $\ell$ on a uniform grid of size 20 on $[0,3]$. We generate $n_M = 100$ independent Monte Carlo trials, and compute the mean of the heterogenized spike estimator $\hat t = \|\hat v\|^2$ and the $e$PCA---or scaled---estimator $\hat t(\hat \alpha) = \hat\alpha \|\hat v\|^2$.

\begin{figure}[h]
\centering
\begin{subfigure}{.33\textwidth}
  \centering
  \includegraphics[scale=0.3]{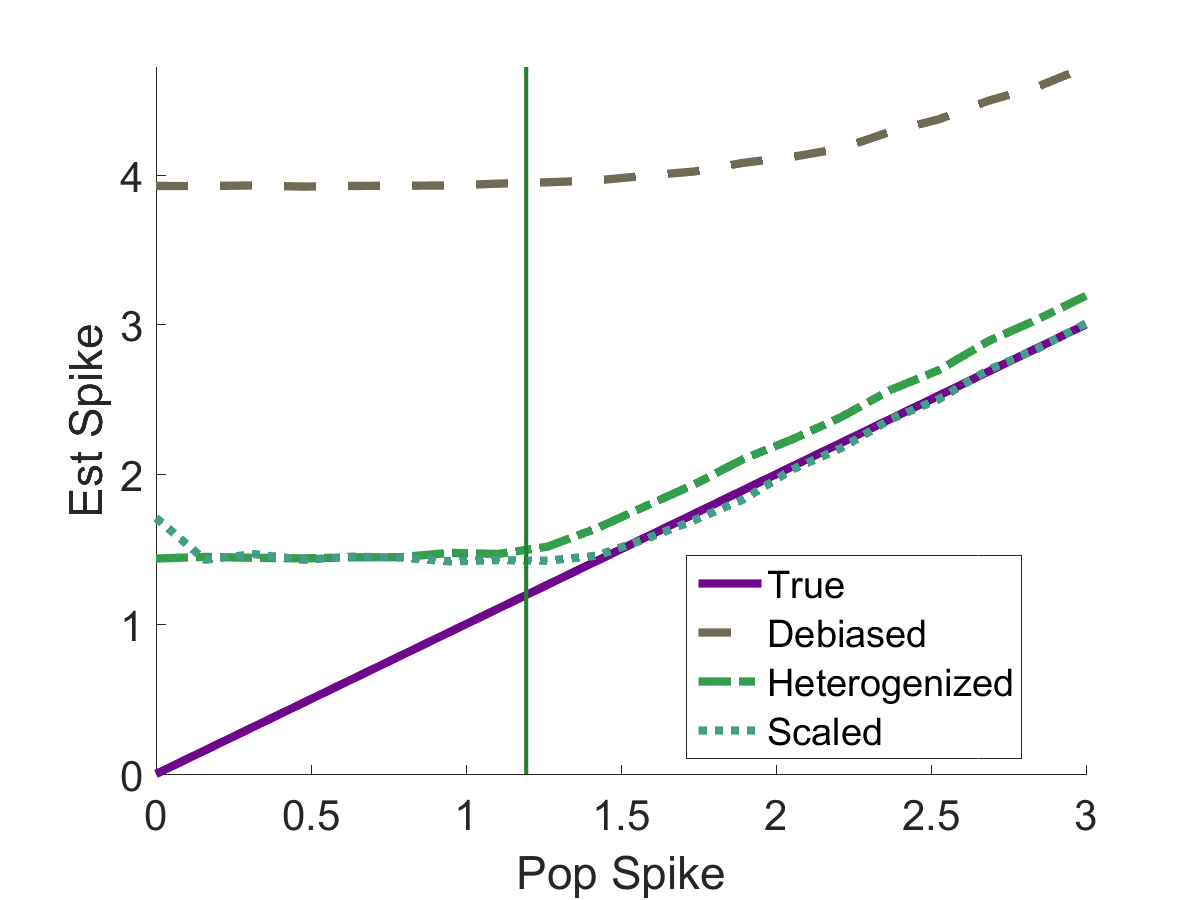}
\end{subfigure}%
\begin{subfigure}{.33\textwidth}
  \centering
  \includegraphics[scale=0.3]{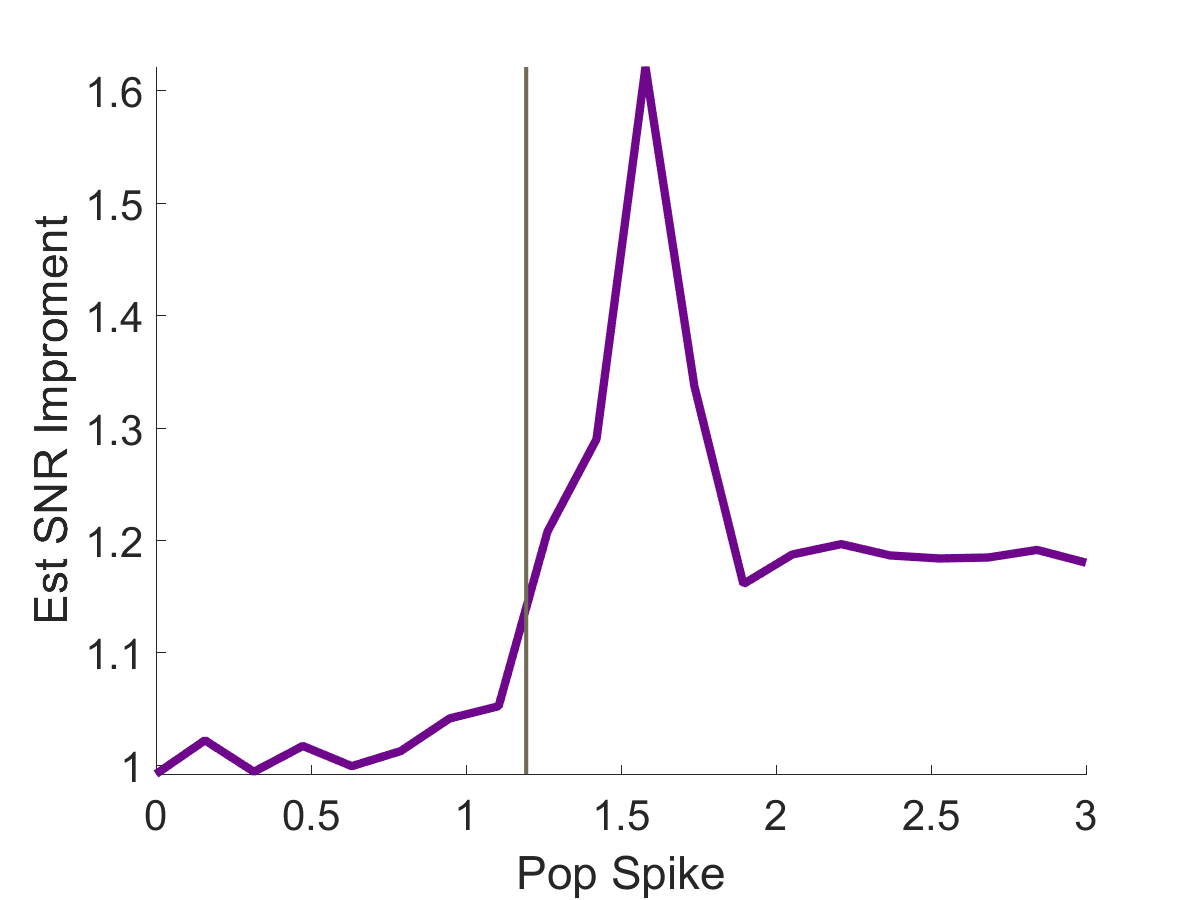}
\end{subfigure}
\begin{subfigure}{.33\textwidth}
  \centering
	\includegraphics[scale=0.3]{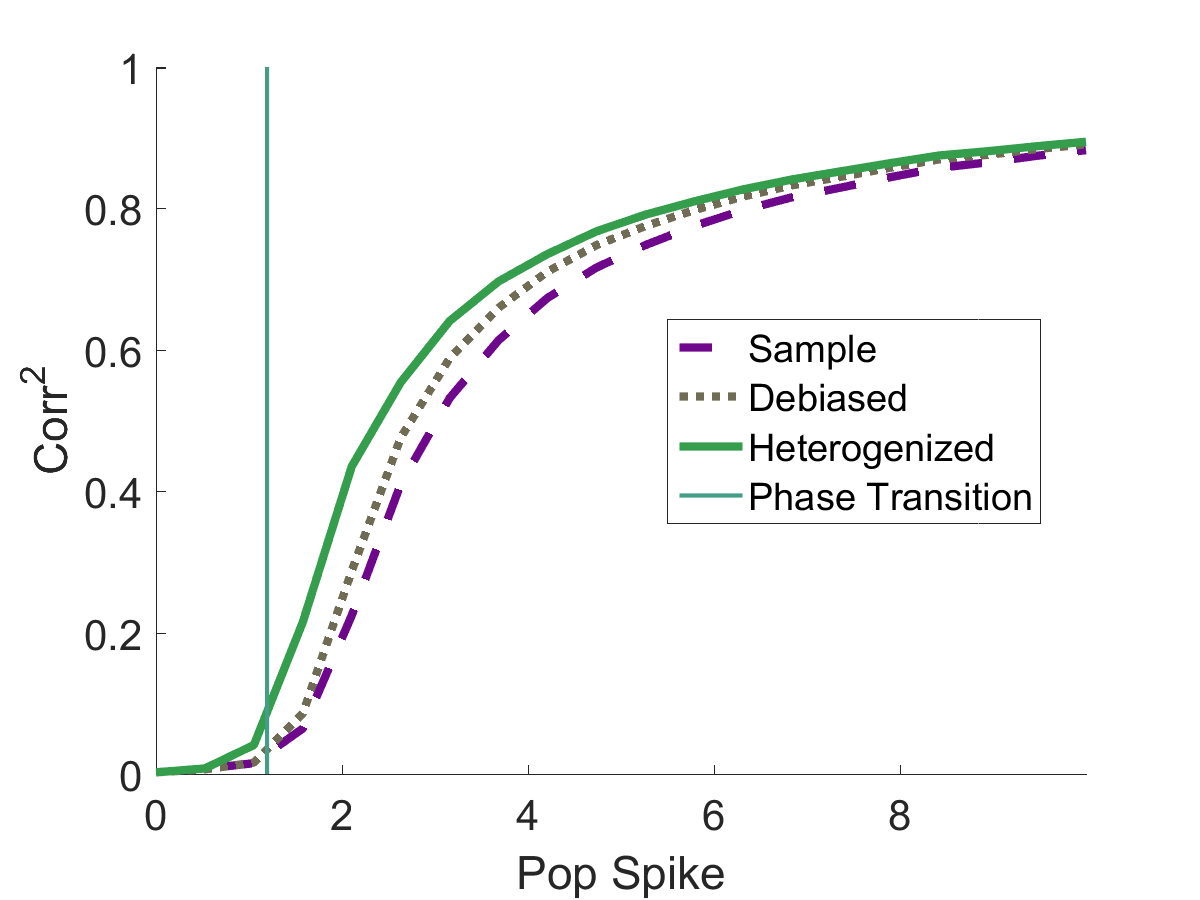}
\end{subfigure}
\caption{Simulation with $e$PCA. Left: Spike estimation; true, debiased, heterogenized, and scaled ($e$PCA). Middle: Estimated improvement in SNR due to homogenization. Right: Squared correlation beween $v$ and leading eigenvector of various covariance estimates; sample, debiased, heterogenized ($e$PCA). Plotted against the spike.}
\label{improved_spike_estim}
\end{figure}

The results displayed in Fig. \ref{improved_spike_estim} (left) show that the $e$PCA/scaled estimator (top eigenvalue of $S_s$) reduces the bias of the heterogenized estimator (top eigenvalue of $S_{he}$) especially for large spikes. Both are much better than the debiased estimator (top eigenvalue of $S_d$). Below the phase transition (vertical line), both estimators have the same approximate value. %We also observed that in very high dimensional regimes where $p > n$, the heterogenized eigenvalue estimates have lower variance compared to debiased eigenvalue estimates.

 We can also define an estimator of the improvement in SNR $\I$, as $\hat \I(\hat \alpha)$. The mean of this estimator over the same simulation is displayed in Fig. \ref{improved_spike_estim} (middle). We observe that it is approximately unity below the PT. This makes sense, because the spike is below the PT both before and after homogenization. The improvement in SNR has a ``jump'' just above the PT, because the spike pops out from the bulk after homogenization. This is where homogenization helps the most. However, $\hat \I$ is not ``infinitely large'', because the signal is detectable in the unhomogenized spectrum, except it is spread across all eigenvalues \citep[see e.g.,][]{dobriban2016sharp}. Finally, $\hat \I(\hat \alpha)$ drops to a lower value, still above unity, and stabilizes. We find this an illuminating way to quantify the improvement due to homogenization.

Finally, we also display the mean of the squared correlation between the true and empirical eigenvectors of various covariance estimators in figure \ref{improved_spike_estim} (right). The predicted PT matches the empirical PT. The $e$PCA eigenvector---top eigenvector of $S_s$---in this case agrees with the eigenvector of the heterogenized covariance matrix $S_{he}$, because both are of rank one. $e$PCA has the highest correlation, and the improvement is significant just above the PT.

\subsection{Homogenization agrees with HWE normalization}
\label{white_hwe}
It is of special interest that for Binomial(2) data, and specifically for biallelic genetic markers such as Single Nucleotide Polymorphisms, our homogenization method recovers exactly the well-known normalization assuming Hardy-Weinberg equilibrium (HWE). In these datasets the entries $X_{ij}$ are counts ranging from 0 to 2 denoting the number of copies of the variant allele of biallelic marker $j$ in the genome of individual $i$. The HWE normalization divides the entries of SNP $j$ by $\sqrt{2 \hat p_j (1-\hat p_j)}$, where $\hat p_j = (2n)^{-1} \sum_i X_{ij}$ is the estimated allele frequency of variant $j$ \citep[e.g.,][p. 2075]{patterson2006population}. It is easy to see that this is exactly the same as our homogenization method assuming that the individual data points $X_{ij}$ are Binomial$(2)$-distributed. 

Previously, the HWE normalization was motivated by a connection to genetic drift, and by the empirical observation that it improves results on observational and simulated data \citep[][p. 2075]{patterson2006population}. Our theoretical results justify HWE normalization. In particular, our Thm. \ref{MP_white} suggests that the Marchenko-Pastur is an accurate null distributions after homogenization. Numerical results also suggest that the approximations to both the MP law and the Tracy-Widom distribution for the top eigenvalue are more accurate than after standardization (data not shown for space reasons). In addition, our result on the improved SNR (Prop. \ref{prop_snr_wh}) suggests that ``signal'' becomes easier to identify after homogenization.

However, in practice we often see similar results with homogenization and standardization. In many SNP datasets, the variants not approximately in HWE---i.e., the variants for which a goodness of  fit test to a Binomial$(2)$ distribution is rejected---are removed as part of data quality control. Therefore, most remaining SNPs have an empirical distribution well fit by a Binomial$(2)$. In such cases standardization and homogenization lead to similar results. 

%\subsection{Multiple exponential families}
%
%Our methodology can handle observations for which the different coordinates have different exponential family noise. For instance, some coordinates may be Gaussian, others may be Binomial, and yet others may be Poisson. This is important in applications where heterogeneous datasets are integrated. For instance, in genomics, there are many different data types, including Gaussian (log-gene expression level), Binomial count (SNP), and Poisson/negative Binomial (RNA-seq).  Our methodology extends because it only depends on the first two moments of the distributions, while our theoretical results only depend on the analytic properties of the coordinate-wise mean-variance map $V$. Thus both the methods and the theory extend to heterogeneous data. 

\section{Denoising}
\label{denoise}

As an application of $e$PCA, we develop a method to denoise the observed data. Formally the goal of denoising is to predict the noiseless signal vectors $X_i = A'(\theta_i)$.  Our model is a random effects model \citep[see e.g.,][]{searle2009variance}, hence we predict $X_i$ using the Best Linear Predictor---or BLP \citep[][Sec. 7.4]{searle2009variance}.  Let $\tilde\E(X|Y) = BY +C$ denote the minimum MSE linear predictor of the random vector $X$ using $Y$, where $B$ is a deterministic matrix, and $C$ is a deterministic vector. This is known under various names, including the \emph{Wiener filter}, see Sec. \ref{contrib}. We will refer to it as the BLP, which is the common terminology in random effects models.
 It is well known \citep[e.g.,][Sec. 7.4]{searle2009variance} that
\begin{equation*}
 B = \Sigma_x \left[\diag[\E A''(\theta)] + \Sigma_x\right]^{-1} \text{ and } 
 C = \diag[\E A''(\theta)] \left[\diag[\E A''(\theta)]+ \Sigma_x\right]^{-1} \E A'(\theta).
\end{equation*}

The BLP depends on the unknown parameters $\Sigma_x$, $\diag[\E A''(\theta)]$,  and $\E[A'(\theta)]$. The standard strategy, known as \emph{Empirical BLP} or EBLP \citep[e.g.,][]{searle2009variance} is to estimate these unknown parameters using the entire dataset, and denoise the vectors $Y_i$ by plug-in: 
\begin{equation*}
 \hat X_i = \hat \Sigma_x \left[\diag[\hat \E A''(\theta)] + \hat \Sigma_x\right]^{-1} Y_i + \diag[\hat \E A''(\theta)] \left[\diag[\hat \E A''(\theta)]+ \hat \Sigma_x\right]^{-1}\bar Y.
\end{equation*}
We will use $e$PCA, i.e., the scaled covariance matrix $S_s$ proposed in \eqref{scaled_cov} to estimate $\Sigma_x$. As before in Sec. \ref{diag_deb}, we will use the sample mean $\bar Y$ to estimate $\E[A'(\theta)]$, and $V(\bar Y)$ to estimate the noise variances $\E A''(\theta)$. However, in principle different estimators could be used. 

For the Poisson distribution, we have
$$\hat X_i=S_{s} \left(\diag[\bar Y] + S_{s} \right)^{-1} \hat Y_i +\diag[\bar Y]\left(\diag[\bar Y] + S_{s} \right)^{-1} \bar Y.$$ 
In some examples there are coordinates where $\bar Y(j)=0$. In our XFEL application this corresponds to pixels where no photon was observed during the entire experiment. This causes a problem because the matrix $\hSigma = \diag[\bar Y] + S_{s}$ may no longer be invertible: $S_s$ is of low rank, while $\diag[\bar Y]$ is also not of full rank. To avoid this problem, we implement a ridge-regularized covariance estimator $\hSigma_\ep = (1-\ep) \hSigma + \ep \cdot \tilde m I_p$ as in \cite{ledoit2004well}, where $\tilde m = \tr \hSigma/p$ and $\ep>0$ is a small constant. Note that $\smash{\tr \hSigma_\ep = \tr \hSigma}$. The ridge-regularized estimator $\smash{\hSigma_\ep}$ has a small bias, but is invertible. In our default implementation we take $\ep = 0.1$%, corresponding to a ``10\% total energy regularization'' in the XFEL application
. Similar results are achieved in our XFEL application for $\ep$ in the range of 0.05--0.2. In new applications we suggest that the user try this range of $\ep$ and choose one based on empirical performance. The same method can be implemented for any exponential family. Another potential solution to the invertibility problem---not pursued here---is to discard the pixels with $\bar Y(j)=0$.

\section{Experiments}

We apply $e$PCA to a simulated XFEL dataset, and an empirical genetics dataset, comparing with PCA.

\subsection{XFEL images}\label{xfel}

We simulate $n_0=70,000$
noiseless XFEL diffraction intensity maps of a lysozyme (Protein Data Bank  1AKI) with Condor \citep{maia2016condor}.  We rescale the average pixel intensity to 0.04 such that shot noise dominates, following previous work \citep[e.g.,][]{Schwander:12}. To sample an arbitrary number $n$ of noisy diffraction patterns, we sample an intensity map at random, and then sample the photon count of each detector pixel from a Poisson distribution whose mean is the pixel intensity. The images are 64 pixels by 64 pixels, so $p=4096$. Figure \ref{fig:xfel} illustrates the intensity maps and the resulting noisy diffraction patterns.

\subsubsection{Covariance estimation} For covariance estimation, we vary the sample size $n$ in the range $3\le \log_{10}(n) \le 5$. We fix the rank of each estimator to be 10, though other choices lead to similar results. The diagonally debiased, heterogenized, and scaled covariance estimates $S_d$, $S_{he}$, $S_s$ each improve on the sample covariance $S$ (Fig. \ref{xfel_err_cov_est}) in MSE. The largest improvement is due to diagonal debiasing, but scaling leads to the smallest MSE. 

\begin{figure}[h]
	\centering
	\begin{subfigure}{.5\textwidth}
		\centering
		\includegraphics[scale= 0.44, trim = 10 0 10 0, clip]{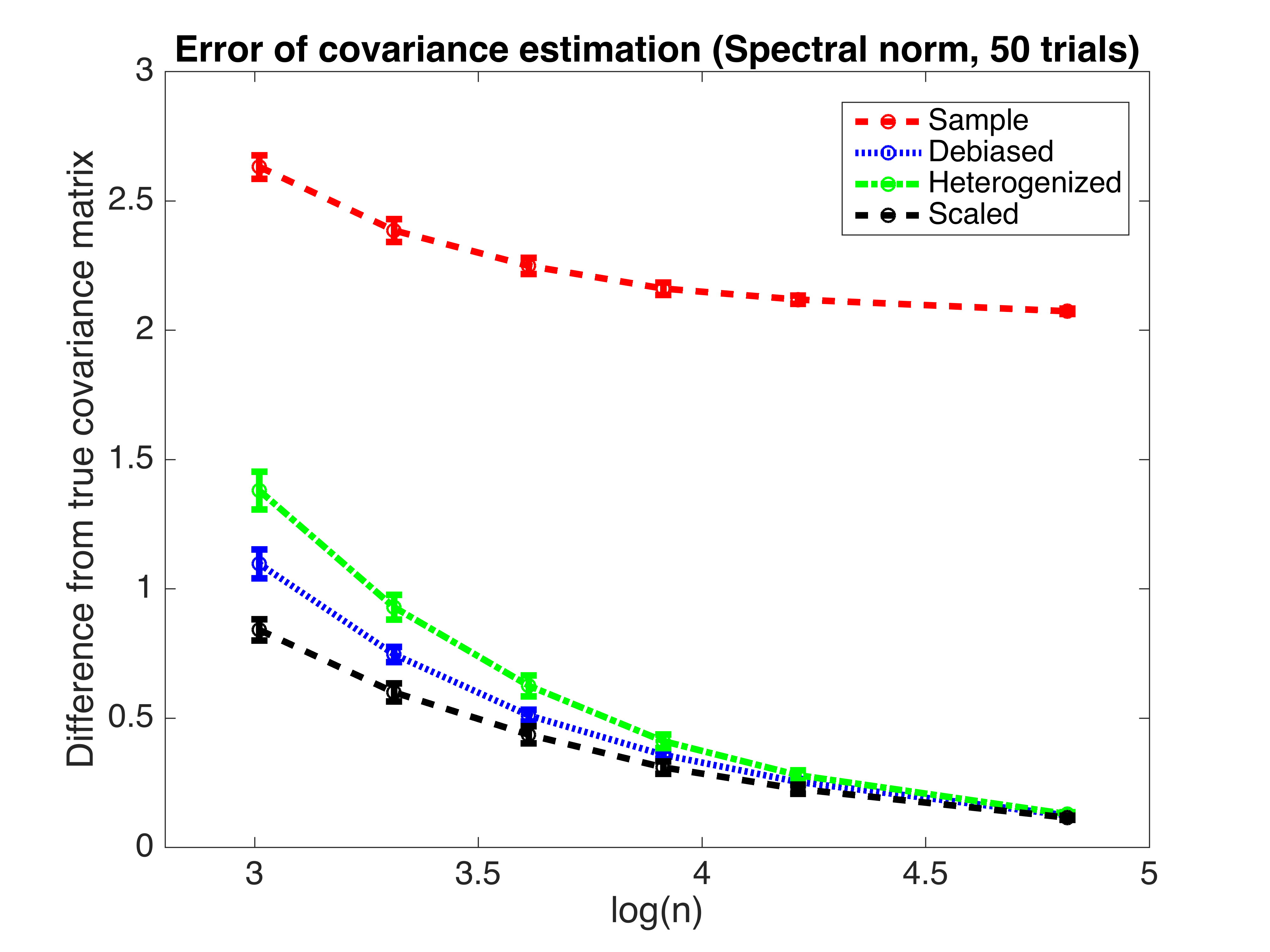} 
	\end{subfigure}%
	\begin{subfigure}{.5\textwidth}
		\centering
		\includegraphics[scale= 0.44, trim = 10 0 10 0, clip]{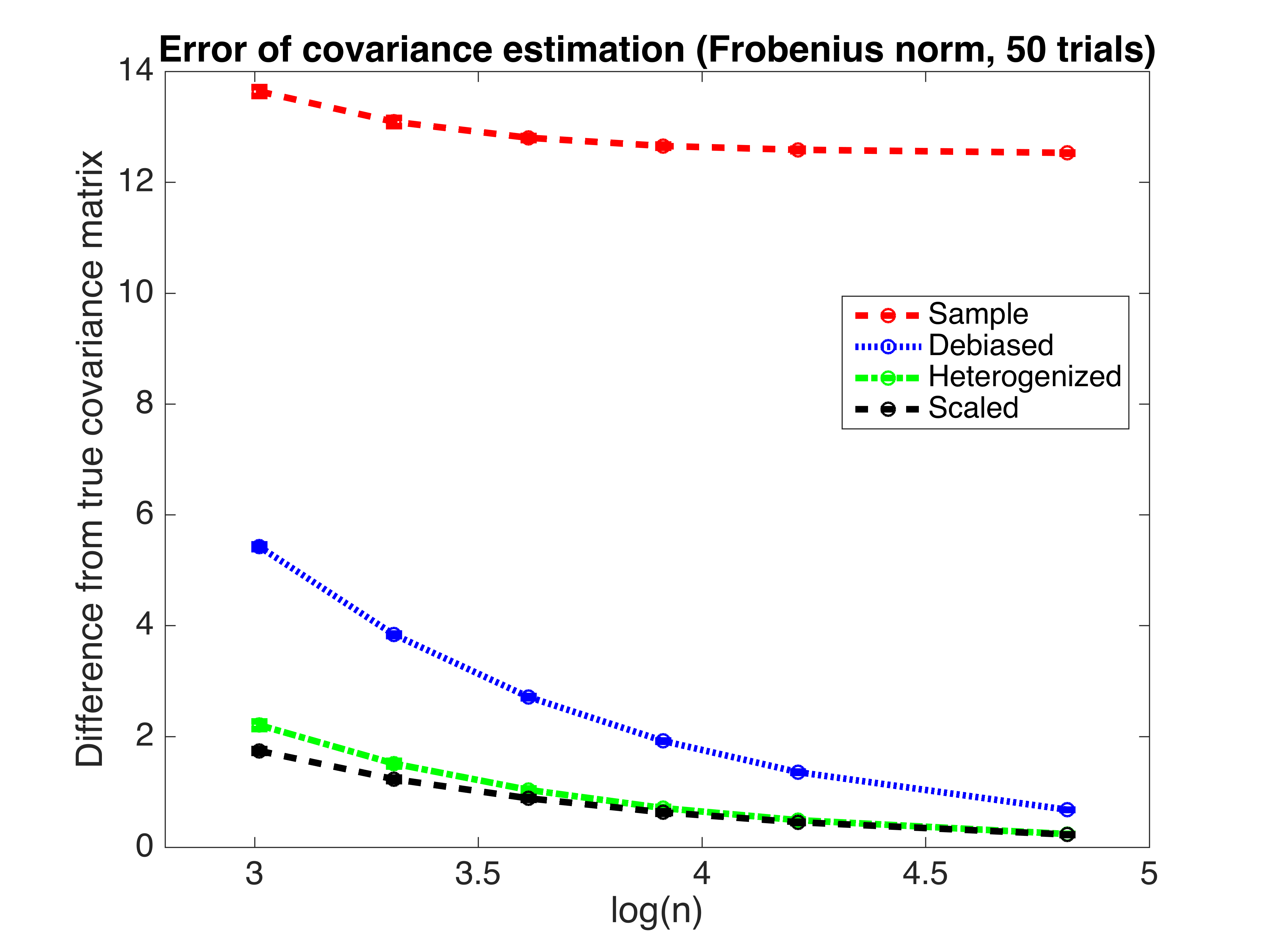} 
	\end{subfigure}
	\caption{Error of covariance matrix estimation, measured as the spectral norm (left) and Frobenius norm (right) of the difference between each covariance estimate (Sample, Debiased, Heterogenized, Scaled) and the true covariance matrix. 
		}
	\label{xfel_err_cov_est}
\end{figure}

\begin{figure}[h]
	\centering
	\includegraphics[width = 1.05\textwidth, trim = 170 330 70 10, clip ]{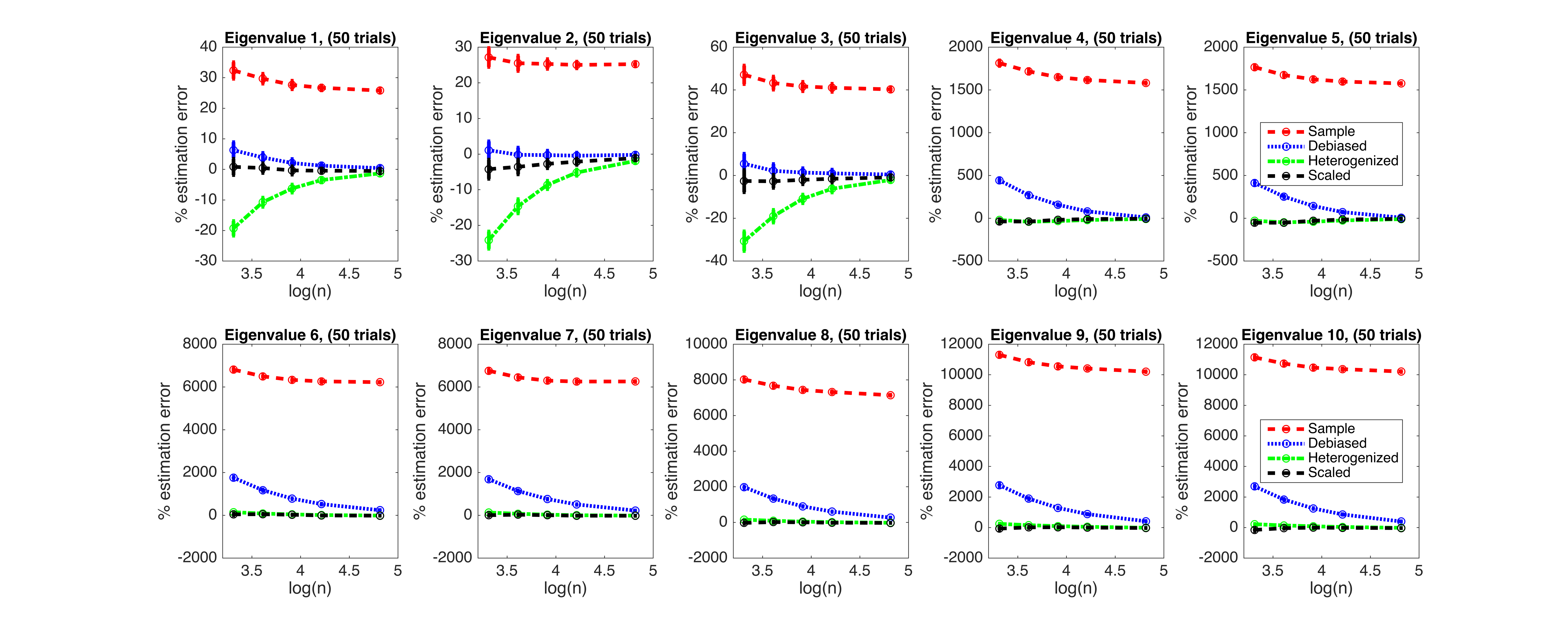} 
	\caption{Error of eigenvalue estimation for the top 5 eigenvalues, measured as percentage error relative to the true eigenvalue, for XFEL data. We plot the mean and standard deviation (as error bars) over 50 Monte Carlo trials.} \label{xfel_err_eval_est}
\end{figure}

Figure \ref{xfel_err_eval_est}  summarizes the error of eigenvalue estimation. The $e$PCA eigenvalues are indeed much closer to the true eigenvalues than the eigenvalues of the debiased or sample covariance matrices $S_d$ or $S$. The estimation error for ePCA eigenvalues is small regardless of sample size. 

We visualize the eigenvectors (or eigenimages) for XFEL diffraction patterns in Figure \ref{eigenpatterns}. The $e$PCA eigenvectors---those of the heterogenized matrix $S_{he}$---accurately estimate two more eigenimages with small eigenvalues than alternative methods. This shows that $e$PCA significantly improves on PCA for covariance estimation in XFEL data.

\begin{figure}[h]
	\centering
	\includegraphics[width = 1.05\textwidth, trim = 90 164 90 55, clip, right]{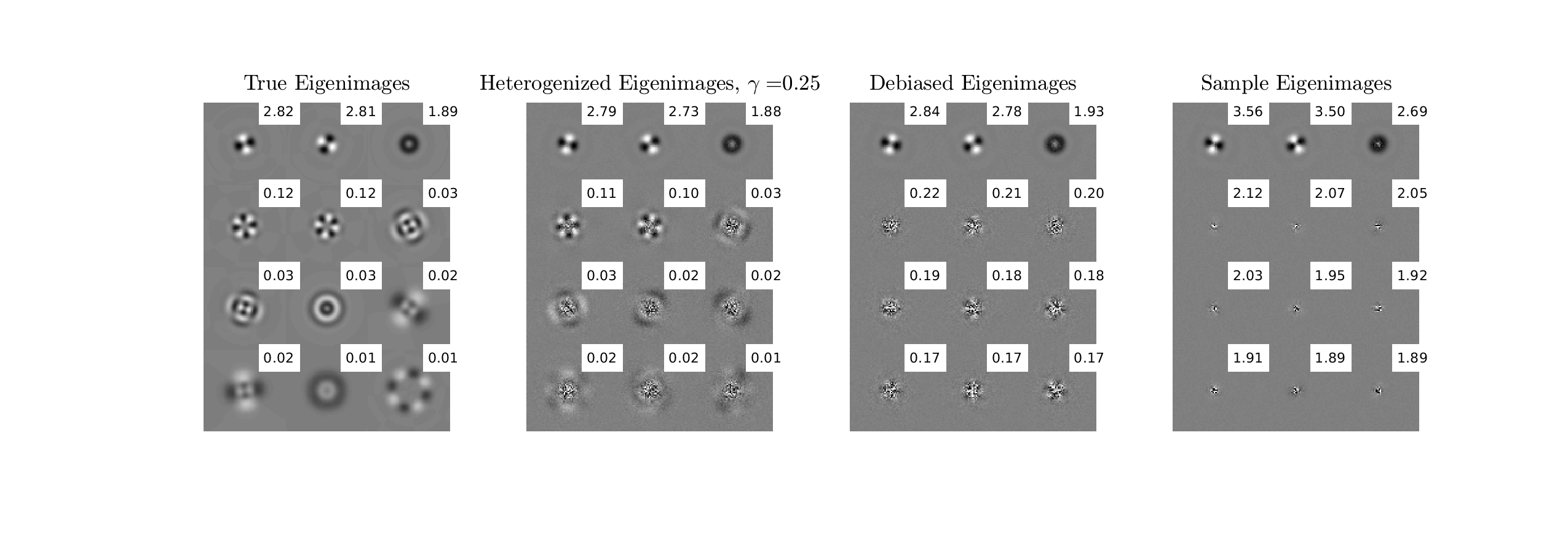}
	\caption{XFEL Eigenimages for $\gamma =  \sfrac{1}{4} $, ordered by eigenvalue} \label{eigenpatterns}
\end{figure}

The $e$PCA/heterogenized eigenvectors 1 to 2 in Figure \ref{eigenpatterns} appear misaligned with the corresponding true eigenvectors.  A likely explanation is that the top eigenvectors have similar eigenvalues, leading to some reordering and rotation in the estimated eigenvectors. 
%This affects the alignment of the estimated eigenvectors and the true eigenvectors. 
Therefore, we also report the error of estimating the overall low-rank subspace, for rank $r=10$, measured as the estimation MSE of the projection matrix $U_rU_r^T$. Other values of $r$ lead to comparable results. Figure \ref{fig:project_fro} clearly shows that the $e$PCA/heterogenized covariance matrix best estimates the low-rank subspace inhabited by the clean data. %Moreover, it is significantly more efficient statistically, improving faster with more data than the two alternative estimates. 

\begin{figure}[h]
	\centering
		\begin{subfigure}[b]{.48\textwidth}
				\centering
				\includegraphics[width = \textwidth, trim = 20 5 15 10, clip]{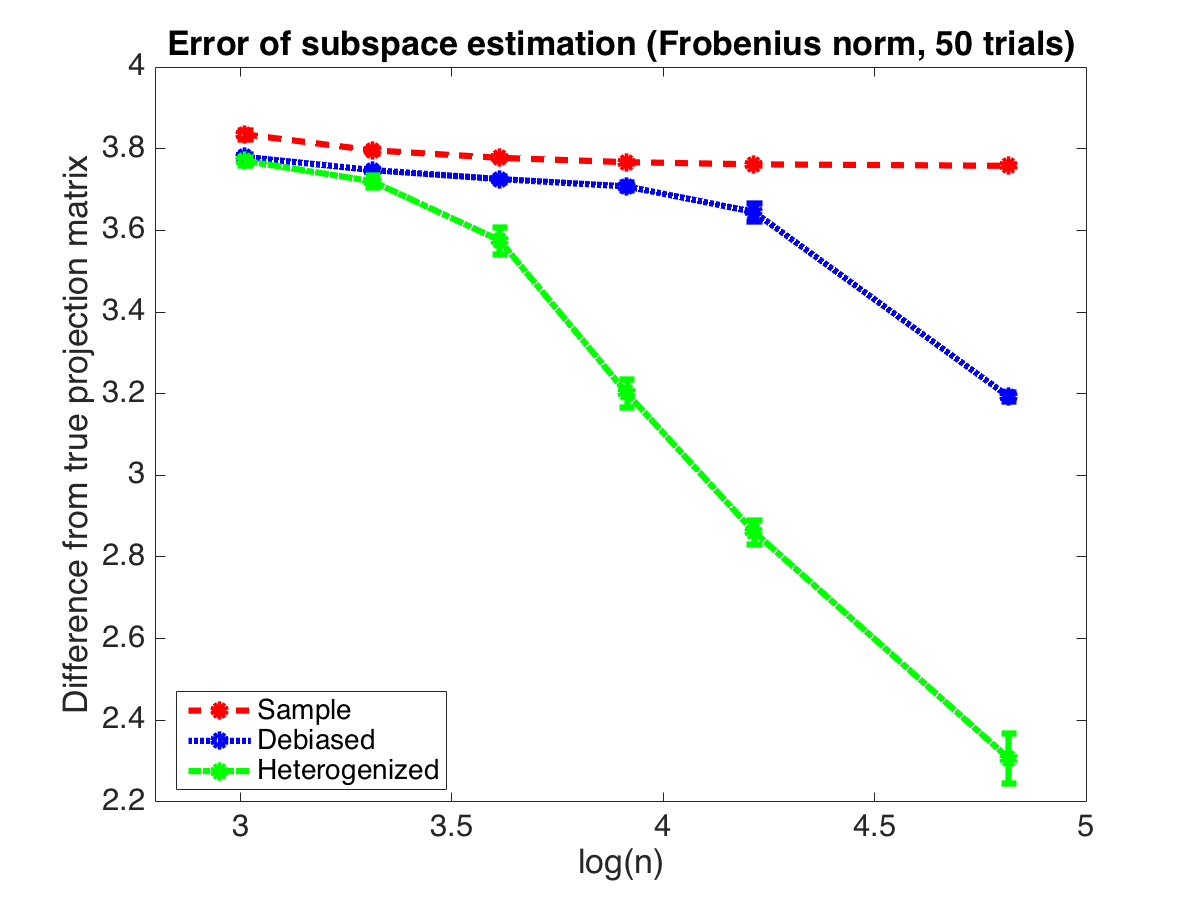} 
				\caption{}\label{fig:project_fro}
		\end{subfigure} 
%	\begin{subfigure}{.48\textwidth}
%		\centering
%		\includegraphics[width = 1.1\textwidth, trim = 10 0 10 10, clip]{denoising_compare2_rank12.png} 
%		\caption{ }\label{denoisingcompare_xfel}
%	\end{subfigure}
		\begin{subfigure}[b]{.48\textwidth}
\centering
\includegraphics[width = 1.1\textwidth]{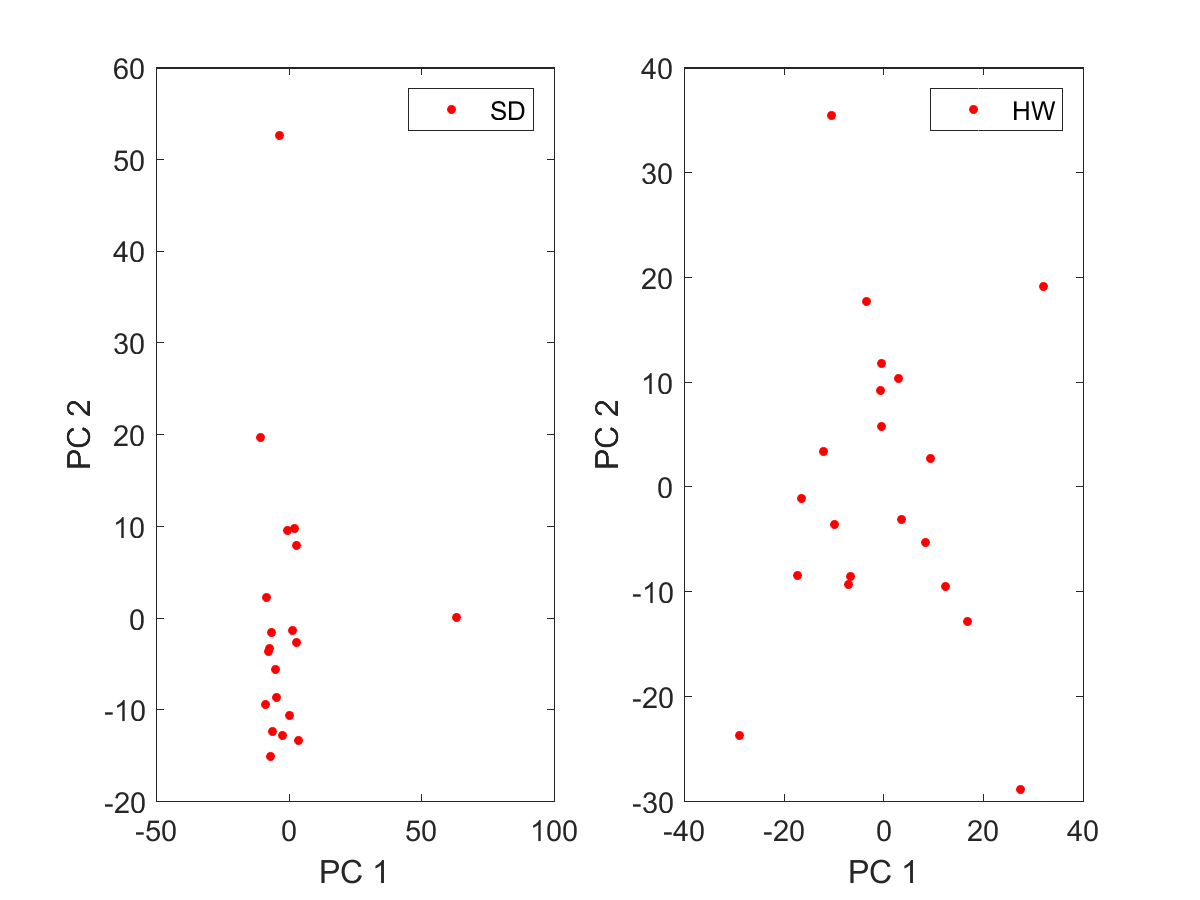}
\caption{}\label{hgdp}
	\end{subfigure}
	\caption{a) Subspace estimation error for XFEL data. We plot the mean and standard deviation (as error bars) over 50 Monte Carlo trials. b) HGDP dataset: PC scores of 20 CEU samples after standardization (SD, left) and homogenization/HWE normalization (HW, right).}
\end{figure}

%As a special note, our final scaling step $S_s$ does not change the eigenvectors compared to heterogenization $S_{he}$, so the heterogenization and scaling are equivalent for all purely eigenvector-dependent measures of error (such as subspace estimation). In those cases heterogenization is identical to $e$PCA.

\subsubsection{Denoising}\label{sec:denoising}

Finally, we report the results of denoising the XFEL patterns. We compare ``PCA denoising'' or ``vanilla projection'', i.e., orthogonal projection onto sample or $e$PCA/heterogenized eigenimages; and EBLP denoising. PCA denoising results in clear artifacts, while the reconstructions after EBLP denoising are always the closest to the clean images (Fig. \ref{denoisingimages_xfel}). In EBLP denoising, our scaled covariance matrix leads to much better results than the sample covariance matrix.  EBLP also does better when measured by reconstruction mean squared error, MSE $:= (pn)^{-1}\sum_{i=1}^n\| \hat{X_i}-X_i \|^2 $.
%$This underscores that it is important to bias-correct before plugging in to the optimal BLP formulas. 
%The EBLP denoiser is also better than both alternatives as measured by mean squared error $pn^{-1}\sum_{i=1}^n\| \hat{Y_i}-X_i \|^2 $, see Fig. \ref{denoisingcompare_xfel}.% Notably, the advantage of EBLP is greater in the high dimensional regime. 

\begin{figure}[t]
	\centering
	\includegraphics[width = 1\textwidth, trim = 70 195 50 150, clip]{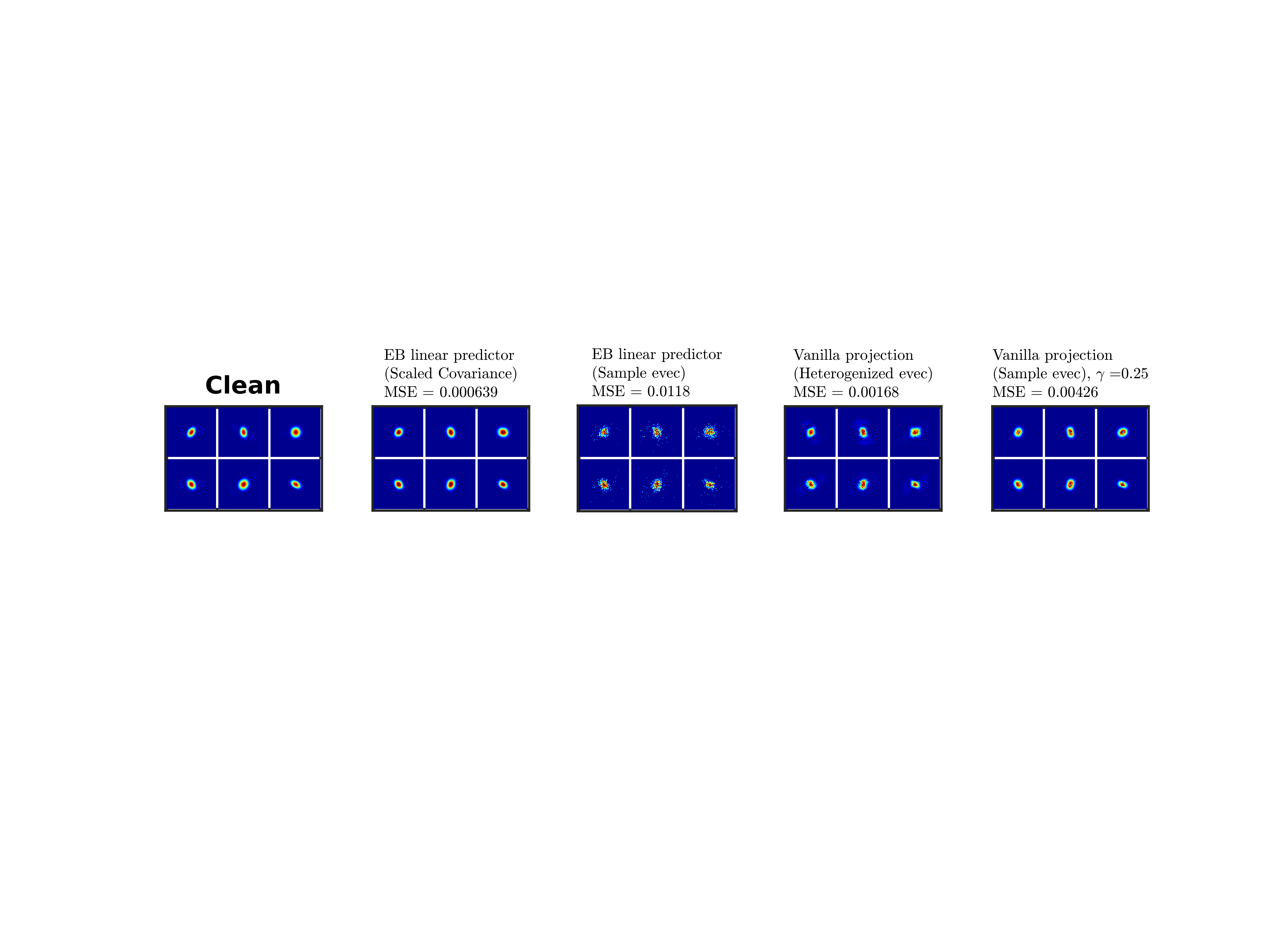} 
	\caption{Sampled reconstructions using the XFEL dataset
		($n = 16,384$; $p = 4096$), fixing the rank of covariance estimates at $r=10$. Color scale of each reconstruction clipped to match that of clean images. } \label{denoisingimages_xfel}
\end{figure}

We also compare $e$PCA to the exponential family PCA method based on alternating minimization proposed by \citet{Collins2001} in Figure \ref{denoisingcompare_xfel}. $e$PCA is faster and recovers the images with higher accuracy, as measured by MSE (see the caption of Figure \ref{denoisingcompare_xfel}). Our experiments with variance stabilizing transforms, such as the Anscombe \citep{anscombe1948} and Freeman-Tukey transforms \citep{freeman1950}, all gave denoising results significantly worse than standard PCA (results not shown due to space limitations). This may be because the known inverse transforms \citep[e.g.,][]{Makitalo2011} are ineffective in the photon-limited regime.

\begin{figure}[t]
	\centering
	\includegraphics[width = .9\textwidth, trim = 70 190 40 155, clip]{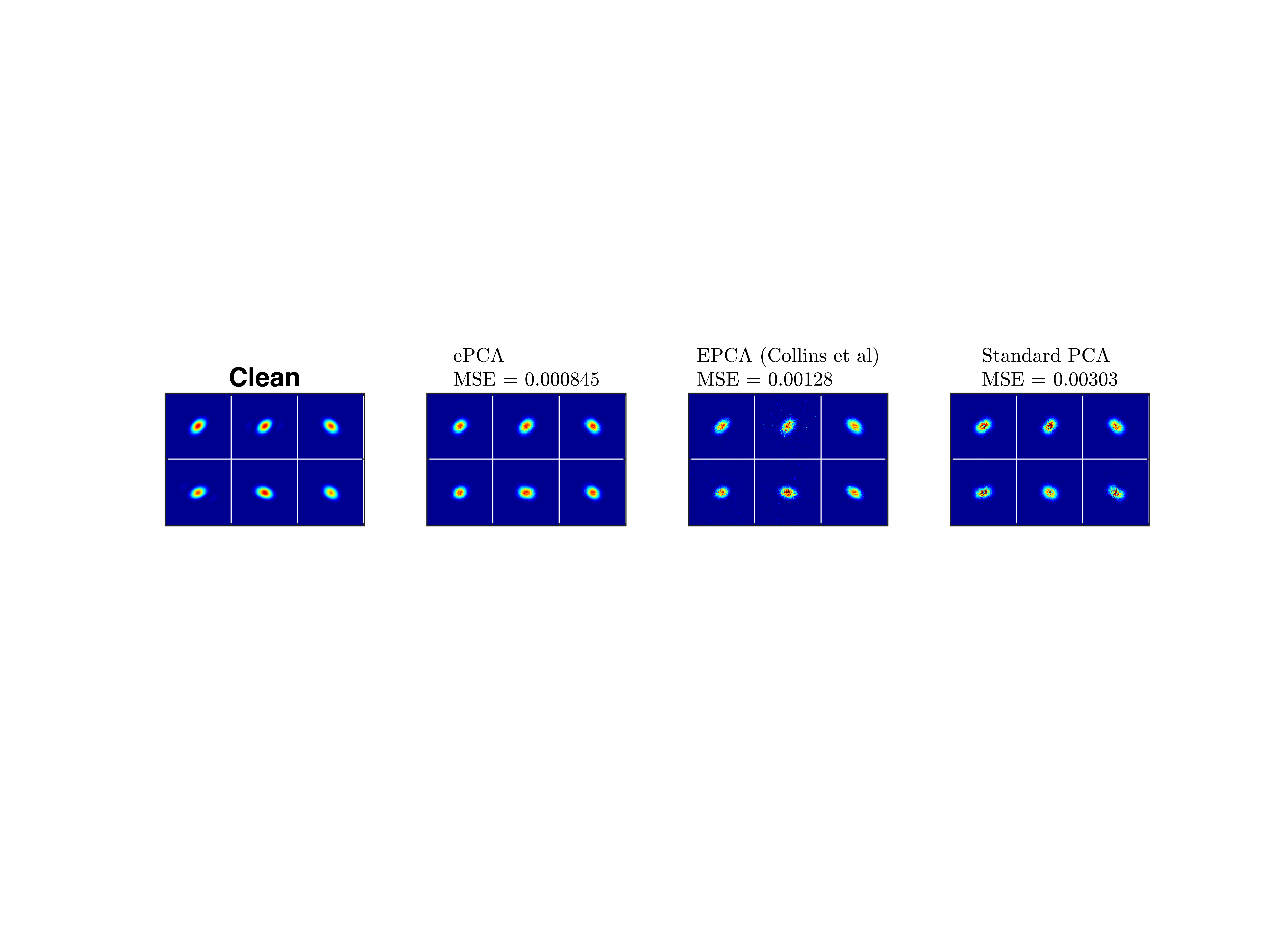} 
	\caption{Comparing various methods' sampled reconstructions of  the XFEL dataset ($n = 1000$; $p = 4096$), fixing the rank estimate for each method to $r=8$. For reference, the MSE for noisy images is 0.0401. We also note that $e$PCA took 13.9 seconds, while \citet{Collins2001}'s exponential family PCA took 10900 seconds, or 3 hours, to finish running on a 2.7 GHz Intel Core i5 processor. } \label{denoisingcompare_xfel}
\end{figure}

\subsection{HGDP dataset}\label{hgdp_sec}

We also apply $e$PCA to a subset of the Human Genome Diversity Project (HGDP) dataset \citep{li2008worldwide}, which contains Single Nucleotide Polymorphism (SNP) markers obtained from human samples. We obtained a homogeneous random set of $n=20$ Caucasian samples from the CEU cohort, typed on $p=120,631$ SNPs. We removed SNPs that showed no variability, with $p' = 107,026$ SNPs remaining.  For each SNP we imputed missing data as the mean of the available samples. We then computed the PC scores starting from two covariance matrices: (1) the one obtained after usual standardization of each feature to have unit norm, and (2) $S_{h}$ obtained by using our homogenization method, which in this case agrees with HWE normalization as defined in e.g., \cite{patterson2006population} (see Sec. \ref{white_hwe}). In Fig. \ref{hgdp} we see that homogenization/HWE normalization apparently leads to a clearer structure in the PC scores than standardization. Two samples on the standardized PC scores appear to be extreme outliers, but our data is a homogeneous random sample and we do not expect outliers. This suggests that standardization is more sensitive to outliers or artifacts. %However, the difference between the two methods is relatively small. 
These results are in line with the existing empirical observations about the superiority of HWE normalization \citep{patterson2006population}.

\section*{Acknowledgements} 
The authors are grateful to Yuval Kluger and Art Owen for helpful comments on an earlier version of the manuscript.  They wish to thank Joey Arthur, Nick Patterson, Kris Sankaran, Joel Tropp, Ramon van Handel, Jingshu Wang, Teng Zhang, and Jane Zhao for valuable discussions. They thank Filipe Maia, Max Hantke, and Benjamin Rose for help with software. 
They thank Patrick Perry for pointing out GLLVMs.
A. S. was partially supported by Award Number
R01GM090200 from the NIGMS, FA9550-12-1-0317 from AFOSR, Simons Foundation Investigator Award and Simons Collaboration on Algorithms and Geometry, and the Moore Foundation Data-Driven Discovery Investigator Award.
E. D. was partially supported by NSF grant DMS-1407813, and by an HHMI International Student Research Fellowship.

%\pagebreak
{\small
\setlength{\bibsep}{0.2pt plus 0.3ex}
\bibliographystyle{plainnat-abbrev}
\bibliography{PACM2015}
}

\appendix
\section{Appendix}
\subsection{Proof of Theorem~\ref{rate_low_dim}}
\label{pf_rate_low_dim}

Let $\mu = \E Y = \E A'(\theta)$ and $B_0=\E Y Y^\top = \Cov Y +\mu \mu^\top =\Sigma_x +\diag[\E A''(\theta)] +\mu\mu^\top$. Let $\|\cdot \|_a$ denote a generic matrix norm, such as the operator norm or the Frobenius norm. By the triangle inequality and the Cauchy-Schwarz inequality
	\begin{align*}
	\E[\V S_d - \sig_x\V_{a}] &= \E \left[\V \frac{1}{n} \sum_{i=1}^n Y_iY_i^\top - \bar Y \bar Y^\top - \diag[V(\bar Y)] - \sig_x \V_{a} \right] \\
	&\leq \E \left[\V \frac{1}{n} \sum_{i=1}^n Y_iY_i^\top - B_0\V_{a}\right] 
	+ \E\left[ \V \bar Y \bar Y^\top - \mu\mu^\top \V_{a} \right] 
	+\E \left[ \V  \diag[V(\bar Y)] - \diag[\E A''(\theta)]  \V_{a} \right] \\
	&\leq \left[\E \V \frac{1}{n} \sum_{i=1}^n Y_iY_i^\top - B_0\V_{a}^2\right]^{\sfrac{1}{2}} 
	+ \E\left[ \V \bar Y \bar Y^\top - \mu\mu^\top \V_{a} \right] 
	+\E \left[ \V  \diag[V(\bar Y)] - \diag[\E A''(\theta)]  \V_{a} \right]
		\end{align*}
We now consider the Frobenius and operator norms separately. For the Frobenius norm, using 		$$\E \left[ \V  \diag[V(\bar Y)] - \diag[\E A''(\theta)]  \V_{\Fr} \right] =  \E \left[ \V V(\bar Y) - \E A''(\theta)  \V \right]$$ and Propositions \ref{THM:mean_bound},  \ref{THM:centering_bound}, and \ref{fro_bound}, we find

$$\E[\V S_d - \sig_x\V_{\Fr}] \lesssim \frac{p}{\sqrt{n}} m_4
	 + \frac{p}{n} + \frac{\| \mu \| \sqrt{p}}{\sqrt{n}}
	 + \frac{\sqrt{p}}{\sqrt{n}}.$$

Now, given that $m_4 = \max_i \E Y(i)^4$ is at least $O(1)$, the second and the last term is of smaller order than the first one. This leads to the bound
$\E[\V S_d - \sig_x\V_{\Fr}] \lesssim  
\sqrt{\frac{p}{n}} \left[
\sqrt{p} \cdot m_4 +\| \mu \| \right].$

	  For the operator norm, using 
$\E \left[ \V  \diag[V(\bar Y)] - \diag[\E A''(\theta)]  \V \right] \le   \E \left[ \V V(\bar Y) - \E A''(\theta)  \V \right]$ 
	  and Propositions \ref{THM:mean_bound},  \ref{THM:centering_bound}, and \ref{tropp_bound}, we find

$$\E[\V S_d - \sig_x\V] \lesssim \sqrt{C(p)}  \frac{(\E\|Y\|^4)^{\sfrac{1}{2}} + (\log n)^3(\log p)^2}{\sqrt{n}}
	 + \frac{p}{n} + \frac{\| \mu \| \sqrt{p}}{\sqrt{n}}
	 + \frac{\sqrt{p}}{\sqrt{n}}.$$
This finishes the proof. 

\subsubsection{Sup-exponential properties}

In this section we establish the sub-exponential property of our random variables. This is needed in the next sections in proving the rates of convergence.

\begin{prop} \label{subexp}
	A random variable $Y\sim p_\theta(y)$ from the exponential family is sub-exponential.
\begin{proof}
The moment generating function of Y is $\E[\exp(tY)] = \exp(A(\theta+t)-A(\theta))$. Since $B$ is differentiable on an open neighborhood of $\theta$, clearly $\E[\exp(tY)]\le e$ for small $t$. Therefore, by the moment generating function characterization of sub-exponential random variables given in (5.16) of  \cite{Vershynin2011}, $Y$ is sub-exponential.
\end{proof}
\end{prop}

In the following proposition, we allow that the prior parameter $\theta$ is random, while requiring that it is bounded.
\begin{prop}
	Let $Y\sim p_\theta(y)$. If $\theta$ is random and supported on a compact interval, then $Y$ is sub-exponential. 
	\begin{proof}
	By the characterization of sub-exponential random variables in (5.16) of  \cite{Vershynin2011}, it is enough to show that $\E[\exp(A(\theta+t)-A(\theta))]\le e$ for small $t$. Suppose $\theta$ is supported on $[a,b]$. Since $A(\theta+t)$ is continuously differentiable in a neighborhood of $\theta$, we have $|A(\theta+t)-A(\theta)| \le C t |A'(\theta)| \le C t \sup_{\theta\in[a,b]}|A'(\theta)|$ for some $C>0$, and for all $t$. Hence
$\E[\exp(A(\theta+t)-A(\theta))]\le \E[\exp(t C \sup_{\theta\in[a,b]}|A'(\theta)|)] \le e.$
The last inequality holds for sufficiently small $t$.
	\end{proof}
\end{prop}

\subsubsection{Auxiliary rates}

Using the sub-exponential properties, we now prove the rates of convergence needed in the proof of Thm. \ref{rate_low_dim} presented in Sec. \ref{pf_rate_low_dim}. Let $K(i) = \sup_{q\ge 1}q^{-1}(\E Y(i)^q)^{1/q}$ be the sub-exponential norm of the $i$-th coordinate of $Y$ \citep[see e.g.,][Sec 5.2.4]{Vershynin2011}. By assumption, these norms are uniformly bounded, so that $K(i) \le K<\infty$ for some universal constant $K$. 

\begin{prop}
\label{THM:mean_bound}
$\E[\V V(\bar{Y}) - \E A''(\theta) \V] \lesssim \frac{\sqrt{p}}{\sqrt{n}} $ up to universal constant factors.
\begin{proof} 
By the Cauchy-Schwarz inequality, $[\E\V V(\bar{Y}) - \E A''(\theta)\V]^2 \le \E[\V  V(\bar{Y}) - \E A''(\theta)\V^2]$. Since the latter quantity decomposes into $d$ mean squared error terms, it is enough to show that each of them is bounded by $C/n$ up to universal constant factors. Now, 
$$\E[{V(\bar Y(i))}-\E{A''(\theta(i))}]^2 \le2 \E[{V(\bar Y(i))}-V(\E{\bar Y(i)})]^2 + 2 \E[V(\E{\bar Y(i)})-\E{A''(\theta(i))}]^2.$$
For the first term, by the Lipschitz property of $V$, and by the definition of $K$, we have 
$$\E[{V(\bar Y(i))}-V(\E{\bar Y(i)})]^2 \le L^2 \E [\bar Y(i)-\E{\bar Y(i)}]^2  = n^{-1} L^2 \Va Y(i)  \le n^{-1} L^2 \E Y(i)^2 \le n^{-1} c L^2 K^2.$$

For the second term, notice that $A''(\theta(i)) = V(\E[\bar Y(i)|\theta(i)])$. Denoting for  convenience $Z=\bar Y(i)$, $\alpha = \theta(i)$, this reads $A''(\alpha) = V(\E[Z|\alpha])$, and thus
$T:=V(\E \alpha)-\E V(\E[Z|\alpha]) = \E\left\{ V(\E \alpha)-V(\E[Z|\alpha])\right\} .$
Hence, by the Cauchy-Shwarz inequality and by the Lipschitz property of $V$, 
$\E T^2 \le \E\left\{ V(\E \alpha)-V(\E[Z|\alpha])\right\}^2 \le L^2 \E(\E \alpha-\E[Z|\alpha])^2.$
Finally, the term  $\E(\E \alpha-\E[Z|\alpha])^2 = \Va (\bar Y(i)|\theta(i)) = n^{-1}  \Va (Y(i)|\theta(i))\le n^{-1} c L^2 K^2$ since $Y(i)$ is sub-exponential with norm at most $K$. Putting together all bounds, we obtain $\E[{V(\bar Y(i))}-\E{A''(\theta(i))}]^2 \lesssim n^{-1}$ up to universal constant factors. By the remark in the beginning of the argument, this finishes the proof. 

\end{proof}
\end{prop}

\begin{prop}
	\label{THM:centering_bound}
	$\E[\V  \mu\mu^\top - \bar{Y}  \bar{Y} ^\top\V_{a} ] \lesssim  \frac{p}{n} + \frac{\| \mu \| \sqrt{p}}{\sqrt{n}} \text{ up to universal constant factors}$, where $\Vert \cdot\Vert_a$ denotes the Frobenius norm or the operator norm.
	\begin{proof}
		Clearly $\V ab^\top\V_{a} = \V a\V \V b\V $. Then 
	\begin{align*}
	\| aa^\top - bb^\top\|_{a}
	&= \| - a(b-a)^\top - (b-a)a^\top -(b-a)(b-a)^\top \|_{a} \\
	&\leq \| (b-a)(b-a)^\top \|_{a} 
	+\| a(b-a)^\top   \|_{a} 
	+ \| (b-a)a^\top  \|_{a} \text{ by the triangle inequality}\\
	&= \| b-a \|^2 +2\| a \| \| b-a \|.
	\end{align*}
	Using this,  by Proposition \ref{THM:mean_bound_2},
$	\E[\|  \mu\mu^\top - \bar{Y}  \bar{Y} ^\top\| _{a}] \leq  \E[\| \mu- \bar{Y} \|^2] +2\E[\|  \mu \| \| \mu- \bar{Y} \| ]
\lesssim  \frac{p}{n} + \frac{\| \mu \| \sqrt{p}}{\sqrt{n}}.$
	\end{proof}

\end{prop}

\begin{prop}\label{THM:mean_bound_2}
We have $\E[\| \bar{Y} - \mu\|^2] \lesssim \frac{p}{n} $ and $\E[\| \bar{Y} - \mu\|] \lesssim \frac{\sqrt{p}}{\sqrt{n}} $ up to universal constant factors.
\begin{proof}
By the Cauchy-Schwarz inequality, $\E[\| \bar{Y} - \mu\|]^2 \le \E[\| \bar{Y} - \mu\|^2]$. Then by the definition of the subexponential norm $K$, we have $\E[{Y(i)}-\E{Y(i)}]^2 \le \E[{Y(i)}]^2 \le c K^2$. Hence
$\E[\| \bar{Y} - \mu\|^2] $ $ = n^{-1}\sum_{i=1}^p \E({Y(i)}-\E{Y(i)})^2 $ $\le n^{-1} cp K^2.$
This finishes the proof.
\end{proof}
\end{prop}

\begin{prop}[Bounding the deviation of the second moment estimator for $Y$: Frobenius norm]\label{fro_bound}
		Let $T_i = \frac{1}{n} \left(Y_iY_i^\top - B \right)$ and $V_n =\sum_{i=1}^nT_i$. Then 
		$\E \left[\V V_n \V_{\Fr}^2 \right] \lesssim \frac{p^2}{n}m_4$.
	\begin{proof}
 Since the $Y_i$ are independent and identically distributed, and $\E T_i = 0$, we have
$$
		\E \left[\V V_n \V_{\Fr}^2 \right] =\E \left[\V \sum_{i=1}^n T_i \V_{\Fr}^2\right] =n\E \V T_1 \V_{\Fr}^2  =  \frac{1}{n}\E( \V Y_1Y_1^\top \V_{\Fr}^2 +\V B\V_{\Fr}^2 -2Tr(Y_1Y_1^\top B))  = \frac{1}{n}(\E(\|Y_1\|^2)^2 - Tr(B^2)).
$$
Now we can bound $\E (\|Y_1\|^2)^2 \le p^2 \max_i \E Y_1(i)^4 \lesssim p^2 m_4$, proving the desired claim.
	\end{proof}
\end{prop}	

\begin{prop}[Bounding the deviation of the second moment estimator for $Y$: Operator norm]\label{tropp_bound}
		Let $T_i = \frac{1}{n} \left(Y_iY_i^\top - B \right)$ and $V_n =\sum_{i=1}^nT_i$. Then 
		\begin{align*}
\E \left[\| V_n \|^2 \right]^{\sfrac{1}{2}} &\leq \sqrt{C(p)} \left\| \E \left[V_n^2 \right]\right\| ^{\sfrac{1}{2}} 
+ \sqrt{C(p)} \cdot \left( \E\left[\max_i \| T_i \|^2 \right]\right)^{\sfrac{1}{2}} \\
&\lesssim \sqrt{C(p)}  \frac{(\E\|Y\|^4)^{\sfrac{1}{2}} + (\log n)^3(\log p)^2}{\sqrt{n}}. 
		\end{align*} 
	\begin{proof}
	The first inequality follows directly from Theorem 5.1 in \cite{Tropp2015}. 
Now we find an explicit expression for the right hand side.  For the first term, since the $Y_i$ are independent and identically distributed, and the $T_i$ are centered,
$ \E V_n^2 = \E (\sum_{i=1}^n T_i)^2 = n  \E T_1^2 
= \frac{1}{n}(\E \|Y_1\|^2 Y_1 Y_1^\top- B^2).$
Since $\E V_n^2$ and $B^2$ are positive semi-definite, so $\|\E \|Y_1\|^2 Y_1 Y_1^\top- B^2\| \le \|\E \|Y_1\|^2 Y_1 Y_1^\top\|$, we have 
$
\left\| \E \left[V_n^2 \right]\right\| \le \frac{1}{n}\| \E( \|Y_1\|^2 Y_1 Y_1^\top)\|.
$ 

Now for any fixed vector $u$ with $\|u\|=1$, $u^\top\E( \|Y_1\|^2 Y_1 Y_1^\top) u = \E \|Y_1\|^2 (u^\top Y_1)^2 \le \E (\|Y_1\|^2)^2$. This gives the first term, $\E \|Y\|^4$.

For the second term,  by the triangle inequality and $(a+b)^2 \le 2(a^2+b^2)$,
$$
\E\left[\max_i \| T_i \|^2 \right] = \frac{1}{n}\E\left[\max_i \| Y_iY_i^T - A \|^2 \right] \\
\le \frac{2}{n}(\E\max_i\| Y_iY_i^T \|^2 + \| B\|^2).$$ 
When taking square roots as required by the theorem statement, the second term in this inequality can be bounded by $\| B \| \le Tr(B) = \E \|Y\|^2 \le (\E \|Y\|^4)^{\sfrac{1}{2}}$. For the first term in the bound, defining $Q_i = \sum_{j=1}^dY_i(j)^2$, we have $\| Y_iY_i^T \|^2 = Q_i^2$, so for $m \geq 2$

$$\E\left[(\max_i Q_i)^2\right] 
\leq \E\left[(\sum_{i=1}^n Q_i^{m})^{\sfrac{2}{m}}\right]
\leq \left(\sum_{i=1}^n \E[Q_i^{m}]\right)^{\sfrac{2}{m}}
= \left(n \E[Q_1^{m}]\right)^{\sfrac{2}{m}},$$

where the second inequality follows from  Jensen's inequality. Choosing $m = \log n$, and then  applying Lemma \ref{Rosenthal} the last term can be upper bounded by

$$ e^2\left(  \left(\E[Q_1^{m}]\right)^{\sfrac{1}{m}}\right)^2
 \lesssim[\E Q_1 + (\log n)^3 (\log p)^2]^2 
\lesssim[Tr(B) + (\log n)^3 (\log p)^2]^2.$$

Finally, we use $Tr(B) = \E \|Y\|^2 \le (\E \|Y\|^4)^{\sfrac{1}{2}}$ again. 
Putting these together leads to the result.
	\end{proof}
\end{prop}	

\begin{lemma}\label{Rosenthal}
Let $Y(1), \cdots, Y(p)$ be independent random variables distributed according to an exponential family $Y(j) \sim p_{\theta(j)}$ for deterministic $\theta(j) \in \R$. Let $Q = \sum_{j=1}^pY(j)^2$. Define $\kappa := \frac{\sqrt{e}}{2(\sqrt{e}-1)} < 1.27$ and let $\eta$ be any value in $(0,1)$. Then for any $m\geq1$ 
$$	(\E[Q^m])^{1/m}  \leq (1+\eta)\E[Q] + C\frac{\kappa}{2}(1+1/\eta) K m^3(\log p)^2 $$
where $C$ is a small constant.
 
 \begin{proof}
By Theorem 8 in \cite{Boucheron2005}, we get the following Rosenthal-type bound:
$$
(\E[Q^m])^{\sfrac{1}{m}} \leq (1+\eta)\E[Q]
+ \frac{\kappa}{2} m (1+1/\eta)\left(\E\left[\left(\max_{j\leq p} (Y(j)^2)\right)^m\right]\right)^{\sfrac{1}{m}}
$$
We proceed to bound the second term on the right hand side:
$$
\E\left[\left(\max_{j\leq p} (Y(j)^2)\right)^m\right]^{\sfrac{1}{m}} 
= \E\left[\max_{j\leq p} Y(j)^{2m}\right]^{\sfrac{1}{m}}
\leq \E\left[\left(\sum_{j\leq p} Y(j)^{2m\log p}\right)^{\sfrac{1}{\log p}}\right]^{\sfrac{1}{m}}
\leq \left(\sum_{j\leq p} \E\left[ Y(j)^{2m\log p}\right]\right)^{\sfrac{1}{m\log p}}
$$
where the last claim follows from Jensen's inequality. This can be further bounded as
$$
p^{\sfrac{1}{m\log p}} \left(\max_{j\leq p} \E\left[ Y(j)^{2m\log p}\right]\right)^{\sfrac{1}{m\log p}} \\
\lesssim K (m \log p)^2.
$$
On the last line, we have used the moments characterization of sub-exponentiality. 
\end{proof} 
\end{lemma}

\subsubsection{Proof of Theorem~\ref{MP_white}}
\label{pf_MP_white}

It is enough to study the singular values of $\mathcal{Y}_w = n^{-\sfrac{1}{2}}\mathcal{Y}_c D_n^{-\sfrac{1}{2}} $, where $\mathcal{Y}_c = \mathcal{Y} - \vec{1} \bar Y^\top$ is the centered data matrix, because $S_{h}+I_p = \mathcal{Y}_w^{\top} \mathcal{Y}_w$. Our strategy to show that $\mathcal{Y}_w$ is well approximated by the noise matrix $n^{-\sfrac{1}{2}}\mathcal{E}$, which is more convenient to study directly.

Indeed, since $n^{-\sfrac{1}{2}}\mathcal{E}$ has independent entries of mean 0, variance $1/n$, and fourth moment of order $1/n^2$, the distribution of the squares of its singular values converges almost surely to the standard Marchenko-Pastur distribution with aspect ratio $\gamma$. Moreover, its operator norm converges  to $1+\gamma^{\sfrac{1}{2}}$ a.s. \citep{bai2009spectral}. 

This implies that the same two properties hold for the auxiliary matrix $\smash{\mathcal{Y}_a = n^{-\sfrac{1}{2}}(\mathcal{Y} - \vec{1} A'(\theta)^\top) D_n^{-\sfrac{1}{2}}}$. Indeed, we can bound the operator norm of  the difference $E=n^{-\sfrac{1}{2}}\mathcal{E}-\mathcal{Y}_a$ as
$$\|E\| = \|n^{-\sfrac{1}{2}} \mathcal{E} (I_p-\diag[A''(\theta)]^{\sfrac{1}{2}}D_n^{-\sfrac{1}{2}})\|  \le \|n^{-\sfrac{1}{2}}\mathcal{E}\| \|I_p-\diag[A''(\theta)]^{\sfrac{1}{2}}D_n^{-\sfrac{1}{2}}\|.$$

Now $\|n^{-\sfrac{1}{2}}\mathcal{E}\| \to 1+\gamma^{\sfrac{1}{2}}$ a.s., and  $\|I_p-\diag[A''(\theta)]^{\sfrac{1}{2}}D_n^{-\sfrac{1}{2}}\| \to 0$ a.s. by Lemma \ref{op_norm_var} presented below. This shows that the spectral distribution and operator norm of $\mathcal{Y}_a $ converge as required. 

Finally, the difference of the homogenized data matrix and the auxiliary matrix has rank one:
$\mathcal{Y}_w-\mathcal{Y}_a = n^{-\sfrac{1}{2}} \vec{1}(A'(\theta)^\top-\bar Y^\top) D_n^{-\sfrac{1}{2}}.$
Therefore $\mathcal{Y}_w$ has the same Marchenko-Pastur limiting spectrum as $\mathcal{Y}_a$. This finishes the proof of Theorem~\ref{MP_white}.  

\begin{lemma}[Convergence of empirical homogenization matrix]
\label{op_norm_var} We have $\|I_p-\diag[A''(\theta)]^{\sfrac{1}{2}}D_n^{-\sfrac{1}{2}}\| \to 0$ a.s.
\end{lemma}

\begin{proof}
Since $|1-x^{\sfrac{1}{2}}| \le |1-x|$ for all $x \ge 0$, it is enough to show that 

$$\max_i \left|1-\frac{A''(\theta(i))}{V(\bar Y(i))}\right| = \max_i \left|\frac{A''(\theta(i))}{V(\bar Y(i))}\right|\left|1-\frac{V(\bar Y(i))}{A''(\theta(i))}\right| \to 0.$$

From the expression on the right, we see that it is enough to show that $\max_i |1-V(\bar Y(i))/A''(\theta(i))| \to 0$ almost surely. To use the Borel-Cantelli lemma, we show how to bound the probability of $V(\bar Y(i))/A''(\theta(i))-1\ge \ep$; the other direction is analogous.  Since we assumed $V$ is Lipschitz continuous with a uniform Lipschitz constant $L$, denoting $\delta = \ep/L$, it is enough to bound the probability that 
$\bar Y(i)- A'(\theta(i))\ge \delta A''(\theta(i))$.
We can write
\begin{align*}
\mathbb{P}\left\{\bar Y(i)- A'(\theta(i))\ge \delta A''(\theta(i))\right\} 
= \mathbb{P}\left\{\frac{\sum_{j=1}^n Y_j(i)}{n} \ge A'(\theta(i)) +\delta A''(\theta(i))]\right\}
\le \E \exp\{t \sum_{j=1}^n Y_j(i)-nt [A'(\theta(i)) +\delta A''(\theta(i))]\}.
\end{align*}

The moment generating function of $Y_j(i)$ is $\exp[A(\theta(i)+t)-A(\theta(i))]$, so the last quantity equals
$\exp[n\{A(\theta(i)+t)-A(\theta(i))-t A'(\theta(i))-t\delta A''(\theta(i))\}].$
For $t$ small enough (depending on $A''$ on a neighborhood of $\theta(i)$), this is less than $\exp[-n\delta A''(\theta(i))/2]$. Since we assumed that $A''(\theta(i))>c$ for some universal constant $c>0$, we get the bound $\exp[-n\delta c/2]$. 

We get a similar upper bound for the probability of deviation in the other direction. We conclude that for some contants $C,c'>0$,  
$\sum_n \Pr(\max_i |1-V(\bar Y(i))/A''(\theta(i))| >\ep) \le C \sum_n n \exp[- c'n ] <\infty.$
hence by the Borel-Cantelli lemma, $\max_i |1-V(\bar Y(i))/A''(\theta(i))| \to 0$ almost surely. This finishes the proof.
\end{proof}

\end{document}